\documentclass[aps,prx,reprint,superscriptaddress,longbibliography,nofootinbib,floatfix]{revtex4-2}
\usepackage[dvipsnames]{xcolor}
\usepackage[colorlinks = true,hyperfootnotes = true,breaklinks = true,citecolor = NavyBlue,linkcolor = NavyBlue,urlcolor = NavyBlue]{hyperref}

\usepackage{amsmath,amssymb,physics,enumerate,graphicx,placeins,bm,adjustbox,amsthm,tabularx,dcolumn,bbold,dsfont}

\usepackage[T1]{fontenc}

\newcommand{\cZ}{\mathcal{Z}}
\newcommand{\cS}{\mathcal{S}}

\newcommand{\cC}{\mathcal{C}}
\newcommand{\cO}{\mathcal{O}}
\newcommand{\cH}{\mathcal{H}}

\newcommand{\cP}{\mathcal{P}}

\newcommand{\cG}{\mathcal{G}}

\newcommand{\cE}{\mathcal{E}}
\newcommand{\cN}{\mathcal{N}}
\newcommand{\cL}{\mathcal{L}}
\newcommand{\cD}{\mathcal{D}}
\newcommand{\cJ}{\mathcal{J}}
\newcommand{\cM}{\mathcal{M}}
\newcommand{\cB}{\mathcal{B}}

\newcommand{\dI}{\mathds{1}}

\newcommand{\ii}{\mathrm{i}}

\DeclareMathOperator{\sgn}{sgn}

\newtheorem{theorem}{Theorem}

\begin{document}

\title{Bath-state (In)dependence of Non-Markovianity Measures in Quadratic Open Systems}

\author{Mirko Consiglio}
\email{mirko.consiglio@um.edu.mt}
\affiliation{Department of Physics, University of Malta, Msida MSD 2080, Malta }

\author{Ilia Khomchenko}
\email{ilia.khomchenko@um.edu.mt}
\affiliation{Department of Physics, University of Malta, Msida MSD 2080, Malta }

\author{Tony J. G. Apollaro}
\email{tony.apollaro@um.edu.mt}
\affiliation{Department of Physics, University of Malta, Msida MSD 2080, Malta }

\date{\today}

\begin{abstract}
    Open quantum system dynamics fail to fulfill the conditions of the Markov approximation when the environment and system timescales are close to each other. Such behavior of these systems can be characterized by several non-Markovianity measures, which are expected to vary with the environment parameters. In this work, we consider a fermionic quadratic open quantum system, establishing bath-state independence of the Breuer--Laine--Piilo and the Lorenzo--Plastina--Paternostro non-Markovianity measures. However, our analytical results reveal that the entanglement-based and divisibility-based Rivas--Huelga--Plenio measures, as well as the entropic Megier--Smirne--Vacchini measure, remain sensitive to the reservoir state, through both its temperature and its chemical potential. Modeling the reservoir as a free-fermionic chain, we derive an exact time-local master equation for a single fermionic site in contact with it. Our numerical results for this model confirm these findings and show that the BLP and LPP measures vanish identically only on the bound-state threshold, where the bath-sensitive measures nevertheless remain finite.
\end{abstract}

\keywords{Open Quantum Systems, Non-Markovianity}

\maketitle

\section{Introduction}

Realistic quantum systems unavoidably interact with their surroundings~\cite{Haroche2006}. To recover from the interaction, the environment needs some time, during which information may continuously flow back from the environment to the quantum system. Such behavior of a quantum system under the influence of its surroundings is often referred to as non-Markovian~\cite{Vega2017}.

Non-Markovian effects arise for a broad scope of quantum systems, including spin chains~\cite{Apollaro2011}, Bose--Einstein condensates~\cite{Haikka2011}, micromechanical oscillators~\cite{Groblacher2015}, chaotic systems~\cite{Zznidaricc2011, Garcia-mata2012}, energy transfer in chemical systems~\cite{Rebentrost2011}, initial system--environment correlations in photonic experiments~\cite{Smirne2011}, continuous-variable quantum key distribution~\cite{Vasile2011}, and quantum metrology~\cite{Chin2012}.

Non-Markovianity is quantified by measures such as the Breuer--Laine--Piilo (BLP)~\cite{Breuer2009}, Lorenzo--Plastina--Paternostro (LPP)~\cite{Lorenzo2013} and Rivas--Huelga--Plenio (RHP)~\cite{Rivas2010} measures~\cite{Rivas2014}, all of which apply directly to two-level systems. Non-Markovianity measures are expected to depend on the parameters of the environment, such as its temperature or chemical potential, since these characterize the environment. This fact has been reported for a range of bosonic systems, including the spin-boson model~\cite{Clos2012}, pure dephasing of a qubit~\cite{Haikka2013}, harmonic-chain reservoirs~\cite{Vasile2014}, and the Caldeira--Leggett model~\cite{Einsiedler2020, Manara2026, Pellitteri2026}, as well as in collisional models with thermal qubit environments~\cite{Man2018}, where the results demonstrated a dependence of non-Markovianity measures on reservoir temperature. For fermionic reservoirs, the question has been studied numerically for a two-level emitter coupled to an XY spin chain, which maps onto free fermions~\cite{Roos2020}: for coupling at the edge of the chain, the trace-distance measure was found to be the same for all environment states that conserve the excitation number and to vanish at a specific detuning, while divisibility could still be broken there. The general mechanism behind these observations, and their scope beyond that model and that class of states, have not been established. Our work aims to fill this gap by showing how non-Markovianity measures respond to the temperature and chemical potential of a fermionic reservoir.

Specifically, we consider a single fermionic site coupled to a fermionic quadratic reservoir through an arbitrary quadratic Hamiltonian. For this system we prove that the BLP measure, based on the trace distance between pairs of evolving states, and the LPP measure, based on the volume of accessible states, are independent of the bath state, and in particular of its temperature and chemical potential. By contrast, the Megier--Smirne--Vacchini (MSV) measure~\cite{Megier2021}, based on the quantum Jensen--Shannon divergence (QJSD), and the entanglement-based (RHPE) and divisibility-based (RHP) Rivas--Huelga--Plenio measures are sensitive to the bath state, through both its temperature and its chemical potential, and we identify a parameter regime in which the BLP and LPP measures vanish identically while all three bath-sensitive measures remain finite.

The paper is organized as follows. Section~\ref{sec:quadratic} describes quadratic system--bath models and presents mathematical proofs of the bath-state independence of the BLP and LPP non-Markovianity measures, as well as of the dependence of the QJSD distinguishability measure on the bath state. Section~\ref{sec:model} is concerned with a particular example of the general model: a chain of non-interacting fermions, acting as a thermal reservoir. Focusing on this fermionic environment, we also derive the master equation governing the evolution of a single fermionic site in contact with this thermal bath. Section~\ref{sec:results} presents numerical results for the spin-chain reservoir at three parameter points: above, on and below the bound-state threshold, together with a map of where the BLP and LPP measures vanish and a scan of the chemical potential. Section~\ref{sec:conclusion} discusses our key findings and outlines their significance. Appendix~\ref{app:dynamical_map} derives the dynamical map for the model of Sec.~\ref{sec:model}. Appendix~\ref{app:NM} derives the explicit forms of the five non-Markovianity measures used in Sec.~\ref{sec:results}. Finally, Appendix~\ref{app:longtime} derives the long-time behavior of the propagator below the threshold.

\section{Quadratic System--Bath Models} \label{sec:quadratic}

We consider a single fermionic site $\hat d$, representing a system, coupled to a reservoir of $K$ sites $\hat c_k$, with the system, reservoir and interaction together forming an arbitrary quadratic Hamiltonian $\hat H$, including pairing terms. The model is generically split as
\begin{equation}
    \hat{H} = \hat{H}_{S} + \hat{H}_{B} + \hat{H}_{SB},
    \label{eq:H}
\end{equation}
where $\hat{H}_S$ is the system Hamiltonian, $\hat{H}_B$ is the bath Hamiltonian, and $\hat{H}_{SB}$ is the interaction Hamiltonian. We write $\hat N_B = \sum_k \hat c^\dagger_k \hat c_k$ for the reservoir particle number, and fix the system basis by taking $\ket{0}$ to be the vacuum and $\ket{1} = \hat d^\dagger\ket{0}$ the occupied state, so that $\rho_{11} = \ev*{\hat d^\dagger\hat d}$ is the occupation of the system site. The model is sketched in Fig.~\ref{fig:quadratic}.

\begin{figure}[t]
    \centering
    \includegraphics[width = \linewidth]{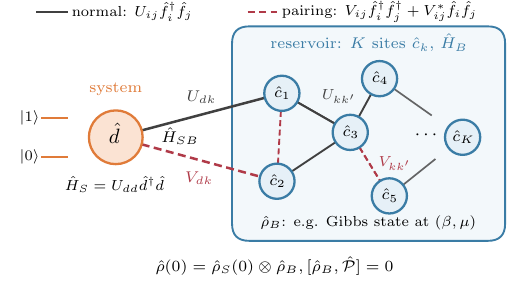}
    \caption{
    Quadratic system--bath model of Sec.~\ref{sec:quadratic}. A single fermionic site $\hat d$, with empty and occupied states $\ket{0}$ and $\ket{1}$, is coupled to a reservoir of $K$ sites $\hat c_k$; system, reservoir and coupling together form an arbitrary quadratic Hamiltonian~\eqref{eq:quadratic}, with normal terms $U_{ij}$ (solid lines) and pairing terms $V_{ij}$ (dashed), $\hat f_i\in\{\hat d,\hat c_1,\dots,\hat c_K\}$. System and reservoir start uncorrelated, and the reservoir state $\hat\rho_B$ need only be parity even, for instance a Gibbs state at inverse temperature $\beta$ and chemical potential $\mu$.
    }
    \label{fig:quadratic}
\end{figure}

Two conditions are imposed on the reservoir. The first is that system and bath are initially uncorrelated, $\hat\rho(0) = \hat\rho_S(0)\otimes\hat\rho_B$. The second is that $\hat\rho_B$ commutes with the parity operator
\begin{equation}
    \hat \cP = (-1)^{\hat N_B} = \prod_{k}(1 - 2\hat c_k^\dagger\hat c_k),
    \label{eq:parity}
\end{equation}
$[\hat\rho_B,\hat \cP] = 0$, i.e.\ $\hat\rho_B$ is block diagonal in the sectors of even and odd fermion number; we call such a state \emph{parity even}. This does not require an even number of fermions: a thermal mixture of even and odd occupations qualifies, and what is excluded is only coherence between the two sectors, such as a superposition of $\ket{0}$ and $\hat c_k^\dagger\ket{0}$. For a parity-even state, expectation values of operators odd in the bath operators vanish: if $\hat \cP\hat O\hat \cP = -\hat O$, then $\ev*{\hat O} = \Tr\{\hat\rho_B\hat \cP\hat O\hat \cP\} = -\ev*{\hat O}$. This holds automatically for the Gibbs state:
\begin{equation}
    \hat \rho_B = \frac{e^{-\beta(\hat{H}_{B}-\mu\hat N_B)}}{\cZ},
    \label{eq:gibbs}
\end{equation}
since $\hat \cP$ maps $\hat c_k \to -\hat c_k$, leaving any quadratic Hamiltonian invariant, so $[\hat H_B,\hat \cP] = [\hat N_B,\hat \cP] = 0$ and hence $[\hat\rho_B,\hat \cP] = 0$. More generally it is guaranteed for any fermionic state by the parity superselection rule~\cite{Wick1952}, so the condition is not a restriction on the class of reservoirs considered. Stationarity, $[\hat H_B, \hat \rho_B] = 0$, is not required by the derivation below, but is what makes $\hat\rho_B$ describe a reservoir \emph{held} at fixed temperature $T$ and chemical potential $\mu$ rather than one that evolves on its own. The commutation of the bath state and Hamiltonian holds for the grand-canonical Gibbs state whenever the bath block is number conserving; if the reservoir carries anomalous terms, one must set $\mu = 0$ and take $\hat\rho_B = e^{-\beta\hat H_B}/\cZ$. Anomalous terms in the system and system--bath blocks are unrestricted throughout.

A generic quadratic Hamiltonian can be defined as
\begin{equation}
    \hat H = \sum_{i, j} U_{ij} \hat f^\dagger_i \hat f_j + \frac{1}{2} \sum_{i, j}(V_{ij} \hat f^\dagger_i \hat f^\dagger_j + V_{ij}^* \hat f_i \hat f_j),
    \label{eq:quadratic}
\end{equation}
where $U = U^\dagger$, the antisymmetry $V = -V^{T}$ follows from the fermionic anticommutation of $\hat f^\dagger_i$ and $\hat f^\dagger_j$, and $\{\hat f_i\} = \{\hat d, \hat c_1,\dots,\hat c_K\}$ collects the system and reservoir sites. This is the Bogoliubov--de Gennes form of a quadratic fermionic Hamiltonian~\cite{Degennes2018, Ring1980}. We now construct the Nambu spinor~\cite{Nambu1960}, $\boldsymbol{\Psi} = (\hat d, \hat c_1, \dots, \hat c_K, \hat d^\dagger, \hat c^\dagger_1, \dots, \hat c^\dagger_K)^T$. A quadratic Hamiltonian can then be rewritten as
\begin{equation}
    \hat H = \frac{1}{2}\boldsymbol{\Psi}^\dagger \cH \boldsymbol{\Psi} + \text{const},~\cH = \begin{pmatrix}
        U & V \\
        -V^* & -U^T
    \end{pmatrix}.
\end{equation}
For a quadratic Hamiltonian, the Heisenberg dynamics are linear $\boldsymbol{\Psi}(t) = e^{-\ii \cH t}\boldsymbol{\Psi}(0)$, whose system row is
\begin{align}
    \hat d(t) &= G(t)\hat d + F(t)\hat d^\dagger + \hat \xi(t), \label{eq:heisenbergC} \\
    \hat \xi(t) &= \sum_k(g_k(t)\hat c_k + \tilde g_k(t)\hat c^\dagger_k).
\end{align}
Since $\hat\xi$ involves only reservoir operators, the canonical anticommutation relations give $G(t) = \ev*{\{\hat d(t),\hat d^\dagger\}}$ and $F(t) = \ev*{\{\hat d(t),\hat d\}}$. Both are fixed entirely by the single-particle matrix $\cH$, so the state of the reservoir enters neither. They are the normal and anomalous propagators of the Bogoliubov--de Gennes problem, defined for $t \geq 0$ with $G(0) = 1$ and $F(0) = 0$, and related to the retarded Green functions by $G^{\rm R}(t) = -\ii\Theta(t)G(t)$ and $F^{\rm R}(t) = -\ii\Theta(t)F(t)$, with $\Theta$ the Heaviside step function~\cite{Ring1980, Economou2006, Degennes2018}.

The expansion of $\hat\xi$ is written in the site basis $\hat c_k$, in which $\hat H$ is defined, rather than in the eigenmode basis that diagonalizes $\hat H_B$. Nothing below requires the bath correlators to be diagonal: the argument uses only that $Z_\beta(t) \equiv \ev*{\hat\xi^\dagger(t)\hat\xi(t)}$ is a real number, independent of the initial state of the system, and never needs its explicit form. For the number-conserving chain of Sec.~\ref{sec:model}, the eigenmode expansion is available, and Appendix~\ref{app:dynamical_map} uses it to evaluate $Z_\beta(t)$ explicitly.

Propagating the coherence $\ev*{\hat d}$ and the population $\ev*{\hat d^\dagger \hat d}$ results in
\begin{equation}
    \ev*{\hat d(t)} = G(t)\ev*{\hat d} + F(t) \ev*{\hat d^\dagger},
    \label{eq:coh}
\end{equation}
and
\begin{widetext}
\begin{align}
    \ev*{\hat d^\dagger(t) \hat d(t)} &= \ev*{[G^*(t)\hat d^\dagger + F^*(t)\hat d + \hat \xi^\dagger(t)][G(t)\hat d + F(t)\hat d^\dagger + \hat \xi(t)]} \nonumber \\
    &= |G(t)|^2\ev*{\hat d^\dagger \hat d} + G^*(t)F(t)\ev*{\hat d^\dagger \hat d^\dagger} + G^*(t)\ev*{\hat d^\dagger \hat \xi(t)} + G(t)F^*(t)\ev*{\hat d \hat d} + |F(t)|^2\ev*{\hat d \hat d^\dagger} \nonumber \\
    &~~~+ F^*(t)\ev*{\hat d \hat \xi(t)} + G(t)\ev*{\hat \xi^\dagger(t) \hat d} + F(t)\ev*{\hat \xi^\dagger(t) \hat d^\dagger} + \ev*{\hat \xi^\dagger(t)\hat \xi(t)} \nonumber \\
    &= |G(t)|^2\ev*{\hat d^\dagger \hat d} + |F(t)|^2\ev*{\hat d \hat d^\dagger} + Z_\beta(t),
    \label{eq:pop}
\end{align}
\end{widetext}
since $\ev*{\hat d^\dagger \hat d^\dagger} = \ev*{\hat d \hat d} = 0$ by the canonical anticommutation relations, while every remaining cross term contains exactly one factor of $\hat\xi$ or $\hat\xi^\dagger$ and therefore factorizes, on the product initial state, into $\ev*{\cdot}_S\ev*{\hat\xi}_B = 0$ by parity. (If the system state is parity even, as for any $\hat\rho_S(0)$ diagonal in the occupation basis, these terms vanish for the system factor as well.) Note that this is the only property of the bath state used anywhere in the derivation. Then we get
\begin{align}
    \rho_{10}(t) &= G(t)\rho_{10}(0) + F(t)\rho_{01}(0), \\
    \rho_{01}(t) &= G^*(t)\rho_{01}(0) + F^*(t)\rho_{10}(0),
\end{align}
and
\begin{align}
    \rho_{11}(t) &= (|G(t)|^2 + Z_\beta(t))\rho_{11}(0) \nonumber \\
    &~~~+ (|F(t)|^2 + Z_\beta(t))\rho_{00}(0), \\
    \rho_{00}(t) &= (1 - |G(t)|^2 - Z_\beta(t))\rho_{11}(0) \nonumber \\
    &~~~+ (1 - |F(t)|^2 - Z_\beta(t))\rho_{00}(0).
\end{align}
Thus the dynamical map, written as a matrix acting on the vectorized state $\hat\rho \mapsto (\rho_{00}, \rho_{10}, \rho_{01}, \rho_{11})^T$, becomes
\begin{equation}
    \cM_t =
    \begin{pmatrix}
        1 - |F|^2 - Z_\beta & 0 & 0 & 1 - |G|^2 - Z_\beta \\
        0 & G & F & 0 \\
        0 & F^* & G^* & 0 \\
        |F|^2 + Z_\beta & 0 & 0 & |G|^2 + Z_\beta
    \end{pmatrix},
    \label{eq:mapC}
\end{equation}
omitting the dependence on time $t$. Complete positivity of Eq.~\eqref{eq:mapC} requires no separate proof: $\cM_t$ is obtained by tracing out the reservoir from unitary evolution of the product state $\hat\rho_S(0)\otimes\hat\rho_B$, and is therefore completely positive and trace preserving by construction.

\subsection{Measures of non-Markovianity} \label{sec:NM}

To characterize the degree of non-Markovianity in the system, there are several measures. We consider five of them: two built on the distinguishability of a pair of evolving states, one on the volume of accessible states, one on entanglement with an ancilla, and one on divisibility of the map. The first is the trace-distance-based BLP measure~\cite{Breuer2009},
\begin{align}
    \cN_\text{BLP} = \max_{\hat\rho_{1,2}(0)} \int_0^\infty \max\left\{0, \partial_t D(\hat\rho_1(t), \hat\rho_2(t))\right\} \dd{t}, \\
    D(\hat\rho_1(t), \hat\rho_2(t)) = \frac{1}{2}\lVert\hat\rho_1(t) - \hat\rho_2(t)\rVert_1.
\end{align}
The second is the LPP measure~\cite{Lorenzo2013}, based on the volume of accessible states,
\begin{align}
    \cN_\text{LPP} &= \int_0^\infty \max\left\{0, \partial_t \det\mathbf{A} \right\} \dd{t}, \\
    (\mathbf{A})_{ij} &= \frac{\Tr{\hat\sigma_i \cM_t(\hat\sigma_j)}}{2},~i, j \in \{x,y,z\}.
\end{align}
Thirdly, the quantum Jensen--Shannon-based entropic MSV measure~\cite{Megier2021, Vacchini2024},
\begin{align}
    \cN_\text{MSV} &= \max_{\hat\rho_{1,2}(0)}\int_0^\infty \max\left\{0,\partial_t\sqrt{J\left(\hat\rho_1(t),\hat\rho_2(t)\right)}\right\}\dd{t}, \\
    J(\hat\rho_1,\hat\rho_2) &= \cS\left(\frac{\hat\rho_1 + \hat\rho_2}{2}\right) - \frac{1}{2} \cS(\hat\rho_1) - \frac{1}{2} \cS(\hat\rho_2),
\end{align}
where $\cS$ is the von Neumann entropy. The fourth one is based on entanglement, the RHPE measure~\cite{Rivas2010},
\begin{align}
    \cN_\text{RHPE} &= \int_0^\infty \left| \partial_t \cE(\hat\rho_{SA}(t)) \right| \dd{t} - \Delta \cE, \\
    \Delta \cE &= \cE(\hat\rho_{SA}(0)) - \cE(\hat\rho_{SA}(\infty)),
\end{align}
where $\cE$ is an entanglement measure. The last one is based on divisibility, the RHP measure~\cite{Rivas2010},
\begin{align}
    \cN_\text{RHP} &= \int_0^\infty \lim_{\varepsilon \to 0^+} \frac{\lVert(\cM_{(t + \varepsilon, t)} \otimes \dI)\ketbra{\Psi}{\Psi}\rVert_1 - 1}{\varepsilon} \dd{t}, \\
    \cM_{(t + \varepsilon, t)} &\equiv \cM_{t + \varepsilon} \cM_t^{-1}.
\end{align}

These five measures are inequivalent and admit only a partial ordering~\cite{Li2018}; comparisons of several measures on a common model have been carried out in Refs.~\cite{Addis2014, Chruscinski2011, Apollaro2014}. They capture different aspects of non-Markovianity, but all vanish for CP-divisible dynamics: if $\cN_\text{RHP} = 0$, the other four measures are zero as well~\cite{Rivas2014, Megier2021}. They are restated in Appendix~\ref{app:NM}, together with their explicit forms for the model of Sec.~\ref{sec:model}.

\subsection{Bath-State Independence of BLP and LPP} \label{sec:independence}

We now show that two of the five measures of Sec.~\ref{sec:NM} are independent of the state of the reservoir --- and hence of its temperature and chemical potential --- for every quadratic Hamiltonian of the form~\eqref{eq:quadratic}. The essential structure of Eq.~\eqref{eq:mapC} is that the bath state enters it in one place only. The map decomposes as
\begin{equation}
    \cM_t(\hat \rho) = \cG_t(\hat \rho) + Z_\beta(t)\Tr{\hat \rho}\hat\sigma^z,
    \label{eq:affine}
\end{equation}
with $\hat\sigma^z = \ketbra{1}{1} - \ketbra{0}{0} = 2\hat d^\dagger\hat d - 1$, and where $\cG_t$ is the linear map obtained from Eq.~\eqref{eq:mapC} by setting
$Z_\beta \to 0$,
\begin{equation}
    \cG_t =
    \begin{pmatrix}
        1 - |F(t)|^2 & 0 & 0 & 1 - |G(t)|^2 \\
        0 & G(t) & F(t) & 0 \\
        0 & F^*(t) & G^*(t) & 0 \\
        |F(t)|^2 & 0 & 0 & |G(t)|^2
    \end{pmatrix}.
    \label{eq:homblock}
\end{equation}
Both pieces are separately trace preserving: $\Tr{\cG_t(\hat \rho)} = \Tr{\hat \rho}$, while $\hat\sigma^z$ is traceless. By Eq.~\eqref{eq:heisenbergC}, $\cG_t$ is built solely from the normal and anomalous propagators $G(t)$ and $F(t)$, which are fixed by the Heisenberg equations and carry no dependence on the bath state. \emph{All} of the $\beta$- and $\mu$-dependence of the dynamics therefore resides in the single scalar $Z_\beta(t)$ multiplying the state-independent shift in Eq.~\eqref{eq:affine}.

\emph{BLP.} The trace distance depends on the pair of inputs only through their difference, and $\hat \rho_1(0) - \hat \rho_2(0)$ is traceless. Hence, for \emph{any} pair of input states,
\begin{equation}
    \cM_t(\hat \rho_1) - \cM_t(\hat \rho_2) = \cG_t\left(\hat \rho_1 - \hat \rho_2\right),
    \label{eq:blpcancel}
\end{equation}
the shift having canceled identically before any optimization is performed. Consequently, $D(\hat \rho_1(t), \hat \rho_2(t))$, the set of times on which it increases, and the maximum over input pairs are all determined by $\cG_t$ alone.

\emph{LPP.} The matrix $\mathbf{A}$ involved in the calculation of the LPP measure is built from $\Tr{\hat \sigma_i \cM_t(\hat \sigma_j)}/2$ with $i,j \in \{x,y,z\}$. The Pauli matrices are traceless, so $\cM_t(\hat \sigma_j) = \cG_t(\hat \sigma_j)$ and $\mathbf{A}$ is precisely the homogeneous block~\eqref{eq:homblock} written in the Bloch basis; the affine shift contributes only to the translation vector of the Bloch ball, to which $\det \mathbf{A}$ is insensitive. Explicitly,
\begin{align}
    &\mathbf{A} =
    \begin{pmatrix}
        \Re(G + F) & \Im(G - F) & 0 \\
        -\Im(G + F) & \Re(G - F) & 0 \\
        0 & 0 & |G|^2 - |F|^2
    \end{pmatrix}, \\
    &\det \mathbf{A} = \left(|G(t)|^2 - |F(t)|^2\right)^2.
\end{align}

We therefore have the following.

\begin{theorem}
    Let the reduced dynamics be of the affine form~\eqref{eq:affine}, with $\cG_t$ independent of the bath state. Then $\cN_\text{BLP}$ and $\cN_\text{LPP}$ are independent of $\hat \rho_B$, and in particular of $\beta$ and $\mu$. This holds for every quadratic Hamiltonian of the form~\eqref{eq:quadratic} whenever the system and bath are initially uncorrelated and the bath is parity even, $[\hat\rho_B,\hat\cP] = 0$ with $\hat\cP$ of Eq.~\eqref{eq:parity}, by Eqs.~\eqref{eq:heisenbergC}--\eqref{eq:mapC}.
    \label{thm:independence}
\end{theorem}

By contrast, $\cN_\text{RHPE}$ and $\cN_\text{RHP}$ do depend on the shift, since both involve $Z_\beta(t)$ explicitly. Their explicit forms are derived in Sec.~\ref{sec:model} and Appendix~\ref{app:NM}.

Note that the argument uses only the affine structure~\eqref{eq:affine} together with the tracelessness of $\hat \rho_1 - \hat \rho_2$ and of the Pauli matrices. It nowhere invokes Gaussianity of $\hat \rho_B$, nor the explicit form of $Z_\beta(t)$, nor stationarity of the reservoir; quadratic Hamiltonians on fermionic parity-even baths are simply the family for which Eqs.~\eqref{eq:heisenbergC}--\eqref{eq:mapC} establish that structure. Nor does it require the optimization over input pairs to be carried out: Eq.~\eqref{eq:blpcancel} holds for every pair separately, so the trace distance of each pair, and hence their maximum, is independent of the bath state. The Gibbs state~\eqref{eq:gibbs} is one instance among many, and temperature independence is the special case of bath-state independence obtained by varying $\beta$ at fixed $\mu$.

For an emitter at the edge of an XY chain and environment states that commute with the excitation number, the bath-state independence of $\cN_\text{BLP}$ was observed numerically in Ref.~\cite{Roos2020}; Theorem~\ref{thm:independence} establishes it for arbitrary quadratic Hamiltonians, pairing terms included, for every parity-even reservoir state, and for $\cN_\text{LPP}$ as well.

Reading the theorem in reverse gives a criterion. If $\cN_\text{BLP}$ or $\cN_\text{LPP}$ is found to depend on the bath state, then the reduced dynamics cannot be of the affine form~\eqref{eq:affine} with a state-independent shift. Since fermionic states are parity even by superselection, only two escape routes remain: the Hamiltonian is not quadratic, containing quartic or higher-order terms; or the system and bath are initially correlated, so that the cross terms dropped after Eq.~\eqref{eq:pop} do not vanish. Either is therefore a \emph{necessary} condition for bath-state-dependent BLP or LPP measures. We have not shown that either is sufficient; whether some interacting or initially correlated model must exhibit temperature-dependent BLP or LPP measures remains an open question.

\subsection{Bath-State Dependence of the MSV Measure} \label{sec:qjsd}

Unlike the trace distance, $J$ is not a function of $\hat\rho_1 - \hat\rho_2$ alone, so the cancellation of the affine shift in Eq.~\eqref{eq:blpcancel} does not take place. We now make this explicit.

Consider the polar pair $\hat\rho_1(0) = \ketbra{1}{1}$, $\hat\rho_2(0) = \ketbra{0}{0}$. The coherence and population blocks of Eq.~\eqref{eq:mapC} do not mix, so both images stay diagonal, with occupations
\begin{equation}
    p_1(t) = |G(t)|^2 + Z_\beta(t),~p_2(t) = |F(t)|^2 + Z_\beta(t),
\end{equation}
since $p_i(t) = \Tr{\cM_t(\hat \rho_i(0))\hat d^\dagger \hat d}$. It is convenient to describe the pair by its mean occupation, $c(t)$, and its gap, $\Delta(t)$,
\begin{align}
    c(t) &= \frac{p_1(t) + p_2(t)}{2} = \frac{|G(t)|^2 + |F(t)|^2}{2} + Z_\beta(t), \label{eq:c} \\
    \Delta(t) &= p_1(t) - p_2(t) = |G(t)|^2 - |F(t)|^2,
\end{align}
so that the bath state enters the mean only: $\Delta$ is built from the propagators alone. The trace distance of this pair is $D = |\Delta|$, which is Eq.~\eqref{eq:blpcancel} for this particular pair. The QJSD, in contrast, is
\begin{equation}
    J_\text{pol} = h(c) - \frac{1}{2} h\left(c + \frac{\Delta}{2}\right) - \frac{1}{2} h\left(c - \frac{\Delta}{2}\right),
\end{equation}
with $h(x) = -x\log_2x - (1 - x)\log_2(1 - x)$ being the binary entropy function.

\begin{theorem}
    Let $0 < p_{1,2} < 1$ and $\Delta \neq 0$. Then (i) $\partial_{|\Delta|} J_\text{pol} > 0$: the QJSD grows with the magnitude of the occupation gap; (ii) $\sgn\partial_c J_\text{pol} = \sgn\left(c - 1/2\right)$: at fixed gap it grows as the pair moves away from half filling, and is stationary only at $c = 1/2$. Since $\partial c/\partial Z_\beta = 1$ and $\partial\Delta/\partial Z_\beta = 0$, $J_\text{pol}$ depends on the bath state except at $c = 1/2$.
    \label{thm:qjsd}
\end{theorem}

\begin{proof}
Throughout, $h'(x) = \log_2((1 - x)/x)$. Note first that $J_\text{pol}$ is invariant under $\Delta \to -\Delta$, which merely exchanges $p_1$ and $p_2$, so it is a function of $|\Delta|$ and we may take $\Delta > 0$ without loss of generality.

(i) Differentiating $J_\text{pol} = h(c) - h(c + \Delta/2)/2 - h(c - \Delta/2)/2$,
\begin{equation}
    \partial_\Delta J_\text{pol} = \frac{1}{4}\left(h'\left(c - \frac{\Delta}{2}\right) - h'\left(c + \frac{\Delta}{2}\right)\right) > 0,
\end{equation}
since $h'(x)$ is strictly decreasing.

(ii) Likewise
\begin{align}
    \partial_c J_\text{pol} &= h'(c) - \frac{1}{2}\left(h'\left(c + \frac{\Delta}{2}\right) + h'\left(c - \frac{\Delta}{2}\right)\right) \nonumber \\
    &= \frac{1}{2}\log_2\left(\frac{(1 - c)^2(c^2 - \frac{\Delta^2}{4})}{c^2\left[(1 - c)^2 - \frac{\Delta^2}{4}\right]}\right),
\end{align}
where $c^2 - \Delta^2/4 = p_1 p_2 > 0$ and $(1 - c)^2 - \Delta^2/4 = (1 - p_1)(1 - p_2) > 0$, so both arguments are positive. Subtracting the denominator from the numerator,
\begin{align}
    &(1 - c)^2(c^2 - \Delta^2/4) - c^2\left[(1 - c)^2 - \frac{\Delta^2}{4}\right] \nonumber \\
    &= \frac{\Delta^2}{4}\left[c^2 - (1 - c)^2\right] \nonumber \\
    &= \frac{\Delta^2}{4}(2c - 1),
\end{align}
so the ratio exceeds $1$ exactly when $c > 1/2$, equals $1$ at $c = 1/2$, and is smaller than $1$ when $c < 1/2$. This is the claim.
\end{proof}

Theorem~\ref{thm:qjsd} concerns a fixed pair of input states. Since $\cN_\text{MSV}$ is a maximum over pairs, pairwise dependence on $Z_\beta$ does not by itself exclude a bath-independent maximum; we close this gap below, by combining a strictly positive lower bound at finite temperature with the exact vanishing of $\cN_\text{MSV}$ in the unital limit.

\paragraph*{Revival criterion.} The pair is fixed by the two numbers $c$ and $|\Delta|$, so
\begin{equation}
    \partial_t J_\text{pol} = \partial_{|\Delta|} J_\text{pol} \partial_t|\Delta| + \partial_c J_\text{pol} \partial_t c.
    \label{eq:Jdot}
\end{equation}
Since $D = |\Delta|$, the trace distance responds to the first term of Eq.~\eqref{eq:Jdot} alone, and grows only when the gap widens. The second term has no counterpart for the trace distance: by Theorem~\ref{thm:qjsd} it is positive whenever
\begin{equation}
    \left(c - \frac{1}{2}\right)\partial_t c > 0,
\end{equation}
that is, whenever the pair drifts away from half filling. The QJSD can therefore revive while the gap shrinks, provided
\begin{equation}
    \partial_c J_\text{pol} \partial_t c > \left|\partial_{|\Delta|} J_\text{pol} \partial_t|\Delta|\right|,
\end{equation}
a possibility excluded for $D$ by construction.

\paragraph*{Infinite temperature.} Since the map is affine, the average of the two evolved occupations is the occupation of the evolved average input, and for the polar pair that average is the maximally mixed state: $c(t) = \Tr{\cM_t(\dI/2)\hat d^\dagger\hat d}$. Now, for a reservoir of $K$ sites, as $\beta \to 0$ at fixed $\mu$, the Gibbs state~\eqref{eq:gibbs} tends to the maximally mixed state of the reservoir, so that $\dI/2\otimes\dI/2^K$ is proportional to the identity on system and reservoir together and is therefore invariant under the unitary evolution. Tracing out the reservoir gives $\cM_t(\dI/2) = \dI/2$, so the reduced dynamics is unital and
\begin{equation}
    c(t) = \frac{1}{2}~\forall t,
\end{equation}
the stationary point of Theorem~\ref{thm:qjsd}(ii). More generally, from Eq.~\eqref{eq:c}, unitality fixes $Z_\beta(t) = (1 - |G(t)|^2 - |F(t)|^2)/2$ in terms of the propagators alone, so that at infinite temperature $J$ no longer depends on the reservoir state for any pair of initial states.

Two consequences follow. First, at $c = 1/2$, the polar-pair QJSD
\begin{equation}
    J_\text{pol} = 1 - h\left(\frac{1 + |\Delta|}{2}\right)
\end{equation}
is a strictly increasing function of the gap alone, so it revives exactly when the trace distance $D = |\Delta|$ of the same pair does. Second, the statement extends from this pair to all of them. Wherever $\cM_t$ is invertible, P-divisibility is equivalent to the trace norm of every Hermitian operator being non-increasing in time~\cite{Chruscinski2011}. For unital qubit dynamics the identity component of such an operator is left untouched, so this reduces to the requirement that the intermediate maps never expand the Bloch ball, which is exactly what $\cN_\text{BLP}$ monitors. The QJSD is an average of relative entropies to the midpoint of the two states, and the relative entropy is contractive under positive trace-preserving maps~\cite{Muller-hermes2017}; hence no pair can revive while the ball contracts, and conversely a pair aligned with a direction in which the ball expands does revive. Therefore
\begin{equation}
    \cN_\text{MSV}(\beta \to 0) = 0
    \iff \cN_\text{BLP} = 0,
    \label{eq:msvinfT}
\end{equation}
which for $F = 0$ means that $|G(t)|$ is non-increasing: in the unital limit the entropic and trace-distance measures collapse onto one another. Note that where $\cN_\text{BLP} > 0$, as in the bound-state regime of Sec.~\ref{sec:results}, $\cN_\text{MSV}$ retains a nonzero value as $\beta \to 0$.

At finite temperature the situation is different. Since $\cN_\text{MSV}$ is a maximum over input pairs, any single pair gives a lower bound, and it suffices to exhibit intervals on which $\partial_t J_\text{pol} > 0$. Take the spin chain of Sec.~\ref{sec:model} at $E = \Omega/2$, $\Gamma = \Omega$, $\mu = 0$, the point of Fig.~\ref{fig:E_0.5_Gamma_1.0}, where $|G(t)|$ decreases monotonically at all times and hence $\cN_\text{BLP} = \cN_\text{LPP} = 0$ identically. Evaluating Eq.~\eqref{eq:Jdot} with the exact $G(t)$ and $Z_\beta(t)$ shows the drift term $\partial_cJ_\text{pol} \partial_t c$ outweighing the shrinking of the gap on several intervals. This gives $\cN_\text{MSV} \geq 9.4 \times 10^{-3}$ at $T = 0.02 \Omega$, a lower bound that decreases monotonically to zero as $T$ grows (Table~\ref{tab:msvfamilies}). Together with the exact vanishing at $\beta \to 0$, Eq.~\eqref{eq:msvinfT}, proved for the chain in Sec.~\ref{sec:master_equation}, this establishes that $\cN_\text{MSV}$ is bath-state dependent.

It is worth being explicit about what the equal-weight trace distance does and does not probe. For a qubit, $D(\hat\rho_1,\hat\rho_2) = |\bm v_1 - \bm v_2|/2$ depends on the states only through the difference of their Bloch vectors, which an affine map transports by its linear part alone; the translation $Z_\beta(t)$, the only place the bath state enters, therefore cannot appear. This is also why $\cN_\text{BLP} = 0$ does not by itself imply P-divisibility: the property equivalent to P-divisibility is monotonicity of the Helstrom matrix norm $\|q\hat\rho_1 - (1 - q)\hat\rho_2\|_1$ for \emph{every} prior $q$, not for $q = 1/2$ alone~\cite{Wissmann2015}. Writing $w_z$ for the longitudinal component of $q\bm v_1 - (1 - q)\bm v_2$, one has $w_z(t) = |G(t)|^2 w_z(0) + (2q - 1)s(t)$ with $s(t) = |G(t)|^2 + 2Z_\beta(t) - 1$, so the bath-dependent term is weighted by $2q - 1$ and is switched off precisely at $q = 1/2$. We have checked that for $q \neq 1/2$ this quantity does revive at the parameters of Fig.~\ref{fig:E_0.5_Gamma_1.0}, where $\cN_\text{BLP} = 0$, consistent with $\cN_\text{RHP} > 0$ there.

The remaining two measures depend on the bath state more directly, since both involve $Z_\beta(t)$ explicitly: the concurrence entering $\cN_\text{RHPE}$ through Eq.~\eqref{eq:concurrence}, and $\cN_\text{RHP}$ through the individual rates $\Lambda^\pm(t)$ of Sec.~\ref{sec:master_equation}, of which only the sum is bath independent. That structure also yields a bath-independent lower bound on $\cN_\text{RHP}$, Eq.~\eqref{eq:rhpfloor}.

\section{A Spin Chain as a Reservoir} \label{sec:model}

We now consider a particular example of the general model of Sec.~\ref{sec:quadratic}, a reservoir with no pairing terms ($V = 0$ throughout), so that the total particle number is conserved, $[\hat H,\hat N] = 0$, the anomalous propagator vanishes ($F(t) = 0$) and the dynamics become analytically tractable: the propagator acquires a scalar spectral representation, and the localized levels responsible for non-Markovian backflow can be located in closed form.

As a concrete example, we take the reservoir to be a semi-infinite spin chain whose terminal site is coupled to the system site, as sketched in Fig.~\ref{fig:setup}. Since we treat the chain as a heat bath, it is characterized by a temperature $T \equiv 1/(k_\text{B}\beta)$ and a chemical potential $\mu$. For simplicity, we set Boltzmann's constant and the reduced Planck constant $k_\text{B} = \hbar = 1$ throughout.
\begin{figure}[t]
    \centering
    \includegraphics[width=\linewidth]{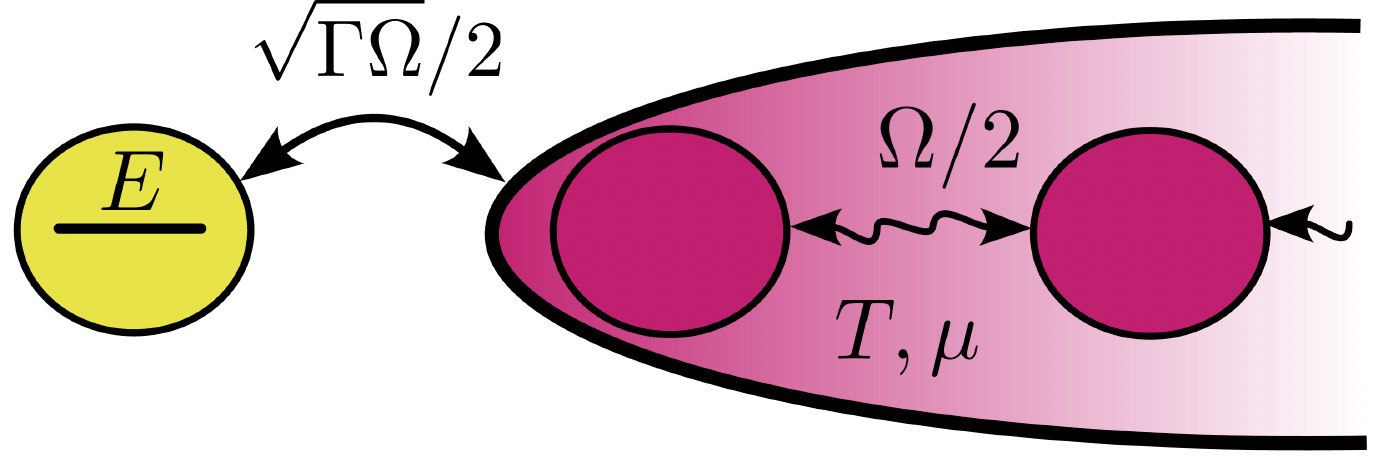}
    \caption{A single fermionic site with on-site energy $E$, coupled to a reservoir modeled as a semi-infinite spin chain. The system--reservoir coupling has strength $\sqrt{\Gamma\Omega}/2$, and the chain has uniform nearest-neighbor hopping $\Omega/2$, giving a band of width $2\Omega$. The reservoir is held at temperature $T$ and chemical potential $\mu$.}
    \label{fig:setup}
\end{figure}

Under the Jordan--Wigner transformation~\cite{Jordan1928, Lieb1961}, an XX spin chain in a transverse field maps onto a chain of spinless free fermions. The transformation can be applied to the bath alone or to the bath together with the system spin at its end; both describe the same spin model and give the same reduced dynamics, but only the latter makes $\hat H$ quadratic, so that Sec.~\ref{sec:quadratic} applies. Ordering the sites from the system, the system becomes the fermionic site $\hat d$ with no string, and since it couples only to the terminal site of the chain, the coupling carries none either. In the fermionic picture, the local magnetic fields become on-site energies, the field applied to the bath spins plays the role of the chemical potential $\mu$, and the exchange coupling between the system and the chain, which transfers excitations, becomes a tunneling term that exchanges particles. The three terms of Eq.~\eqref{eq:H} then read
\begin{align}
    \hat{H}_{S} &= E\hat{d}^{\dagger}\hat{d},
    \label{eq:HS} \\
    \hat{H}_{B} &= -\frac{\Omega}{2}\sum_{k = 1}^{K-1} (\hat{c}^{\dagger}_{k}\hat{c}_{k + 1} + \hat{c}^{\dagger}_{k + 1}\hat{c}_{k}), \\
    \hat{H}_{SB} &= \frac{\sqrt{\Gamma\Omega}}{2}(\hat{d}^{\dagger}\hat{c}_{1} + \hat{c}^{\dagger}_{1}\hat{d}),
    \label{eq:HSB}
\end{align}
where $E$ is the on-site energy of the system. Equations~\eqref{eq:HS}--\eqref{eq:HSB} define a resonant-level, or Fano--Anderson, model~\cite{Anderson1961, Fano1961}: a single discrete level hybridized with a continuum. The choice of a uniform chain is not restrictive, since any bath spectral density can be mapped onto a one-dimensional chain with suitably chosen hoppings~\cite{Chin2010}. The reservoir can be diagonalized via the canonical transformation
\begin{equation}
    \hat{a}_{q} = \sqrt{\frac{2}{K + 1}}\sum_{k = 1}^{K}\sin(qk)\hat{c}_{k},~q = \frac{\pi n}{K + 1},~n = 1,\dots,K,
    \label{eq:aq}
\end{equation}
with wavenumbers $q$, which results in
\begin{align}
    \hat{H}_{B} &= \sum_{q}\omega_{q}\hat{a}^{\dagger}_{q}\hat{a}_{q}, \\
    \hat{H}_{SB} &= \sum_{q}(\tau_{q}\hat{d}^{\dagger}\hat{a}_{q} + \tau^{*}_{q}\hat{a}^{\dagger}_{q}\hat{d}),
\end{align}
where the dispersion relation is $\omega_{q} = -\Omega\cos(q)$, and the tunnel couplings are $\tau_{q} = \sqrt{\Gamma\Omega/(2K + 2)}\sin(q)$.

The reservoir is initialized in the Gibbs state~\eqref{eq:gibbs} where, equivalently, $\hat N_B = \sum_q \hat a^\dagger_q \hat a_q$ in the eigenmode basis. Note that a uniform field on the bath spins commutes with $\hat N_B$, and so can be absorbed into $\mu$. The dynamics preserve the total number of fermions since $[\hat{H},\hat N] = 0$, where $\hat N = \hat N_{S} + \hat N_B$ with $\hat N_S$ being the particle number operator of the system. Consequently, the populations and the coherences of $\hat{\rho}_{S}$ evolve independently, and the reduced dynamics is fixed by the population
\begin{equation}
    C(t) = \Tr{\hat{\rho}_{S}(t)\hat{d}^{\dagger}\hat{d}},
    \label{eq:Cdef}
\end{equation}
which satisfies $0 \leq C(t) \leq 1$, together with the coherence $\ev*{\hat{d}(t)}$. Note that operators without a time argument are taken at $t = 0$.

We can compute Eq.~\eqref{eq:Cdef} by solving the equation of motion for $\hat{d}(t)$. The solution for $t > 0$ can be completely described by the propagator
\begin{equation}
    G(t) = \bra{0}\hat{d}(t)\hat{d}^{\dagger}\ket{0},
    \label{eq:Gdef}
\end{equation}
where $\ket{0}$ is the vacuum state. This coincides with the definition below Eq.~\eqref{eq:heisenbergC} for $F = 0$: $\bra{0}\hat d(t)\hat d^\dagger\ket{0} = \bra{0}\{\hat d(t),\hat d^\dagger\}\ket{0} - \bra{0}\hat d^\dagger\hat d(t)\ket{0} = G(t)$, where the second term vanishes because, for $[\hat H,\hat N] = 0$, $\hat d(t)$ is a linear combination of annihilation operators and so annihilates the vacuum. Let $\varphi(\omega)$ denote the probability distribution of excitation energies associated with the state $\hat{d}^{\dagger}\ket{0}$~\cite{Mitchison2018},
\begin{equation}
    \varphi(\omega) = \sum_{n}|\bra{n}\hat{d}^{\dagger}\ket{0}|^{2} \delta(\omega - \varepsilon_{n}),
    \label{eq:varphi-def}
\end{equation}
where $\ket{n}$ is a single-particle eigenstate of $\hat{H}$ with energy $\varepsilon_{n}$. Since $\hat{H}\ket{0} = 0$ and $\hat{d}^{\dagger}\ket{0}$ lies in the one-particle sector, inserting the completeness relation $\sum_{n}\ketbra{n}{n}$ into Eq.~\eqref{eq:Gdef} gives
\begin{equation}
    G(t) = \int e^{-\ii \omega t}\varphi(\omega) \dd{\omega},
    \label{eq:G-fourier}
\end{equation}
so that the normalization of $\varphi$ is equivalent to $G(0) = 1$,
\begin{equation}
    \int \varphi(\omega) \dd{\omega} = 1,~\lim_{\Gamma\to0}\varphi(\omega) = \delta(\omega - E).
    \label{eq:varphi-norm}
\end{equation}

The effect of the reservoir on the system is determined by the spectral density
\begin{equation}
    \cJ(\omega) = \sum_{q}|\tau_{q}|^{2}\delta(\omega - \omega_{q}).
    \label{eq:Jdef}
\end{equation}
Following Ref.~\cite{Mitchison2018}, in the limit $K \to \infty$ the level spacing $\Delta q = \pi/(K + 1)$ tends to zero, $q$ becomes a continuous variable on $[0,\pi]$, and $\sum_{q} \to (K + 1)/\pi \int_{0}^{\pi}\dd{q}$, so that Eq.~\eqref{eq:Jdef} reduces to the Newns spectral density~\cite{Newns1969}
\begin{equation}
    \cJ_\text{N}(\omega) = \frac{\Gamma}{2\pi}\sqrt{1 - \frac{\omega^{2}}{\Omega^{2}}} \Theta(\Omega - |\omega|).
    \label{eq:newns}
\end{equation}
The associated impurity self-energy, obtained as $|\sqrt{\Gamma\Omega}/2|^{2}$ times the surface Green function of a semi-infinite chain of hopping $\Omega/2$~\cite{Economou2006}, reads
\begin{equation}
    \Sigma(\omega) = \frac{\Gamma}{2\Omega} \times
    \begin{cases}
        \omega - \ii\sqrt{\Omega^{2}-\omega^{2}}, & |\omega| \leq \Omega,\\
        \omega - \operatorname{sgn}(\omega)\sqrt{\omega^{2}-\Omega^{2}}, & |\omega| > \Omega,
    \end{cases}
    \label{eq:selfenergy}
\end{equation}
i.e.\ the retarded branch $\omega \to \omega + \ii0^{+}$, fixed by $\Sigma(\omega) \to \Gamma\Omega/(4\omega)$ as $|\omega| \to \infty$. Its real and imaginary parts inside the band are $\Re\Sigma = \Gamma\omega/(2\Omega)$ and $\Im\Sigma = -\pi\cJ_\text{N}(\omega)$, and the excitation-energy distribution~\eqref{eq:varphi-def} evaluates to
\begin{equation}
    \varphi(\omega) = \frac{\cJ_\text{N}(\omega)}{[\omega - E - \Re\Sigma(\omega)]^{2} + [\pi\cJ_\text{N}(\omega)]^{2}}.
    \label{eq:varphi}
\end{equation}
Equation~\eqref{eq:varphi-def} is a sum over all single-particle eigenstates of $\hat{H}$. When $|E| > \Omega - \Gamma/2$, one such eigenstate is a localized level lying outside the band, at the root $\omega_{b}$ of $\omega = E + \Sigma(\omega)$ with $|\omega_{b}| > \Omega$. It contributes an isolated delta function to Eq.~\eqref{eq:varphi-def}, and hence a non-decaying term to the propagator,
\begin{align}
    G(t) &= \sum_{b}R_be^{-\ii \omega_{b}t} + \int_{-\Omega}^{\Omega} e^{-\ii \omega t}\varphi(\omega) \dd{\omega}, \label{eq:G-bound} \\
    R_b &= |\bra{\omega_b}\hat d^\dagger \ket{0}|^2 = [1 - \Sigma'(\omega_{b})]^{-1},
\end{align}
with $\int_{-\Omega}^{\Omega} \varphi \dd{\omega} + \sum_{b} R_b = 1$, which is Eq.~\eqref{eq:varphi-norm}. Outside the band $\Sigma'(\omega) = \Gamma\left(1 - |\omega|/\sqrt{\omega^2 - \Omega^2}\right)/(2\Omega)$, and since $|\omega| > \sqrt{\omega^2 - \Omega^2} > 0$ there, $\Sigma'(\omega_b) < 0$. Since $\Sigma(\pm\Omega) = \pm\Gamma/2$ and $\omega - \Sigma(\omega)$ is strictly increasing outside the band, a root $|\omega_b| > \Omega$ exists if and only if $|E| > \Omega - \Gamma/2$, and it is unique on each side. Hence $R_b\in(0,1)$ automatically, as a residue must be if it is to be the weight $|\bra{\omega_b}\hat d^\dagger\ket{0}|^2$: it tends to $1$ as $\Gamma \to 0$, where the level decouples, and to $0$ at the threshold $|E| = \Omega - \Gamma/2$, where it merges into the band.

Localized levels outside a band are a standard feature of models with bounded spectra~\cite{Zhang2012}, and are responsible for the non-decaying term in Eq.~\eqref{eq:G-bound}. For $\Gamma < 2\Omega$ at most one such level exists. The threshold $|E| = E_c \equiv \Omega - \Gamma/2$ reproduces a known spin-chain result: under $\Omega = 2J$, $\Gamma = 2J(J_0/J)^2$ and $E = 2h$, it becomes $h/J = 1 - (J_0/J)^2/2$, the parabola separating the regimes of Ref.~\cite{Apollaro2011}.

Completing the square in the denominator of Eq.~\eqref{eq:varphi} recovers the Lorentzian form of Ref.~\cite{Mitchison2018},
\begin{equation}
    \varphi(\omega) = \frac{1}{1 - \Gamma/\Omega} \frac{\cJ_\text{N}(\omega)}{(\omega - E')^{2} + \Gamma'^{2}/4},
    \label{eq:varphi-lorentzian}
\end{equation}
with the shifted energy and decay rate
\begin{equation}
    E' = \frac{1 - \Gamma/2\Omega}{1 - \Gamma/\Omega}E,~\Gamma' = \frac{\sqrt{1 - \Gamma/\Omega - E^{2}/\Omega^{2}}}{1 - \Gamma/\Omega}\Gamma.
\end{equation}
Here $\Gamma'$ is real only for $\sqrt{E^{2} + \Gamma\Omega} < \Omega$. In the window $\sqrt{\Omega^{2} - \Gamma\Omega} < |E| < \Omega - \Gamma/2$ the denominator of Eq.~\eqref{eq:varphi-lorentzian} acquires real roots, but these lie outside the band and correspond to virtual rather than bound states. Equation~\eqref{eq:varphi} holds for all parameters; in particular it remains regular at $\Gamma = \Omega$, where Eq.~\eqref{eq:varphi-lorentzian} is singular.

We now define the noise kernel $\Phi(t)$ as
\begin{equation}
    \Phi(t) = \int e^{-\ii\omega t} \cJ_\text{N}(\omega)f(\omega) \dd{\omega},
\end{equation}
where $f(\omega) = [e^{\beta(\omega - \mu)} + 1]^{-1}$ is the Fermi--Dirac distribution of the reservoir. The noise kernel receives no contribution from the localized levels, since $\cJ_\text{N}$ has support only within the band. Solving the Heisenberg equations of motion (see Appendix~\ref{app:dynamical_map} and Ref.~\cite{Mitchison2018}), the solution of Eq.~\eqref{eq:Cdef} can then be represented as
\begin{align}
    C(t) &= |G(t)|^{2}C(0) + Z_\beta(t), \\
    Z_\beta(t) &= \int_{0}^t \int_{0}^t G^*(t')\Phi(t' - t'')G(t'') \dd{t''} \dd{t'}. \label{eq:Zdouble}
\end{align}

\subsection{The Master Equation} \label{sec:master_equation}

To model the dynamics of our system, we use a master equation, originally derived for a double quantum dot in contact with two reservoirs~\cite{Mitchison2018}. In our case, the master equation governing the dynamics of the open system has the following form
\begin{align} \label{eq:master_equation}
    \partial_t \hat{\rho}_S(t) = -\ii[\tilde H_S(t), \hat{\rho}_S(t)] &+ \Lambda^+(t) \cD[\hat{d}^\dagger]\hat{\rho}_S(t) \\ \nonumber
    &+ \Lambda^-(t) \cD[\hat{d}]\hat{\rho}_S(t),
\end{align}
where the dissipator is defined as
\begin{equation}
    \cD[\hat A]\hat \rho = \hat A\hat \rho \hat A^\dagger - \frac{1}{2}\{\hat A^\dagger \hat A, \hat \rho\},
\end{equation}
and $\tilde H_S(t) = H(t)\hat d^\dagger\hat d$ is the renormalized system Hamiltonian, distinct from the bare $\hat H_S$ of Eq.~\eqref{eq:HS}; note that $C(t) = \rho_{11}(t)$. $H(t)$ is a real function, and we also define $\gamma(t) \equiv \Lambda^+(t) + \Lambda^-(t)$. Assuming $G(t) \neq 0$, the propagator $G(t)$ can be related to the real functions $\gamma(t)$ and $H(t)$~\cite{Mitchison2018} by
\begin{equation}
    \ii H(t) + \frac{1}{2}\gamma(t) = -\frac{\partial_t G(t)}{G(t)}.
    \label{eq:gammaH}
\end{equation}
The gain and decay rates $\Lambda^\pm(t)$ are given by
\begin{align}
    \Lambda^+(t) &= \partial_t Z_\beta(t) + \gamma(t) Z_\beta(t), \label{eq:lambdaplus} \\
    \Lambda^{-}(t) &= \gamma(t) - \Lambda^{+}(t) \label{eq:lambdaminus}.
\end{align}
Equations~\eqref{eq:lambdaplus} and~\eqref{eq:lambdaminus} follow from the map of Appendix~\ref{app:dynamical_map}: differentiating Eq.~\eqref{eq:rho11} and using $\gamma(t) = -\partial_t\ln|G(t)|^2$ gives $\partial_t\rho_{11} = -\gamma(\rho_{11} - Z_\beta) + \partial_t Z_\beta = \Lambda^+\rho_{00} - \Lambda^-\rho_{11}$, which is the population sector of Eq.~\eqref{eq:master_equation}. The master equation is therefore exact rather than a Born--Markov approximation, and the rates $\Lambda^\pm(t)$ may take either sign: $\cM_t$ remains completely positive at all times, but it is CP-divisible ($\cM_{t + \varepsilon}\cM_t^{-1}$ completely positive for all $t,\varepsilon > 0$) if and only if $\Lambda^\pm(t) \ge 0$ throughout~\cite{Rivas2014, Hall2014}.

One has $H(0) = E$ and $\gamma(0) = 0$, but $H(t)$ is not constant: a short-time expansion of Eq.~\eqref{eq:G-fourier} in the cumulants of $\varphi$ gives $H(t) = E + E\Gamma\Omega t^{2}/8 + O(t^{4})$, and at long times $H(t)$ approaches the bound-state energy $\omega_b$ when a localized level exists. Only in the wide-band limit $\Omega \to \infty$ at fixed $\Gamma$ does $H(t)$ reduce to $E$ for all $t$. None of the five measures considered below depends on $H(t)$: they involve the map only through $|G(t)|$ and $Z_\beta(t)$, and in $\cN_\text{RHP}$ the phase cancels from $\det\chi^{(c)}_\varepsilon$ at first order in $\varepsilon$ (see Appendix~\ref{app:RHP}). Moreover, the assumption that $G(t) \neq 0$ is exactly the condition for $\cM_t$ to be invertible, $\det\cM_t = |G(t)|^4$.

The relation between the measures is visible in the rates themselves. By Eq.~\eqref{eq:gammaH}, $\gamma(t) = -\partial_t\ln|G(t)|^2$, so $|G(t)|$ increases precisely when $\gamma(t) = \Lambda^+(t) + \Lambda^-(t) < 0$. Both $\cN_\text{BLP}$ and $\cN_\text{LPP}$ therefore register non-Markovianity exactly on the intervals where the \emph{sum} of the rates is negative, whereas $\cN_\text{RHP}$ responds whenever either rate alone is negative (see Appendix~\ref{app:RHP}). Consequently, $\cN_\text{BLP} > 0 \implies \cN_\text{RHP} > 0$, but not conversely: the divisibility-based measure captures different aspects of non-Markovianity, as also shown in Ref.~\cite{Li2018}. It also explains the bath-state independence established in Sec.~\ref{sec:independence} without appeal to Eq.~\eqref{eq:affine}: by Eqs.~\eqref{eq:lambdaplus}--\eqref{eq:lambdaminus} the sum $\Lambda^+ + \Lambda^-$ is $\gamma(t)$, in which $Z_\beta$, and hence $f(\omega)$, has canceled, while the individual rates retain it.

The same decomposition yields a bath-independent floor for the divisibility measure. Since $\max\{0,-a\} + \max\{0,-b\} \geq \max\{0,-(a + b)\}$ for any real $a$ and $b$, Eq.~\eqref{eq:lambdaminus} applied to Eq.~\eqref{eq:RHPresult} gives
\begin{equation}
    \cN_\text{RHP} \geq \cN^{(0)}_\text{RHP} \equiv \int_0^\infty \max\{0,-\gamma(t)\} \dd{t}
    \label{eq:rhpfloor}
\end{equation}
for every state of the reservoir. By Eq.~\eqref{eq:gammaH}, $\cN^{(0)}_\text{RHP} = \int_0^\infty\max\{0,\partial_t\ln|G|^2\} \dd{t}$ is fixed by the propagator alone. The bound is saturated whenever $f(\omega)$ is constant over the band, $f(\omega) = \bar f$ for $|\omega| \leq \Omega$: Eqs.~\eqref{eq:noise_app} and~\eqref{eq:Zbound} then give $Z_\beta(t) = \bar f [1 - |G(t)|^2]$, so that $\Lambda^+ = \bar f \gamma$ and $\Lambda^- = (1 - \bar f) \gamma$ change sign together with $\gamma$. This covers the infinite-temperature limit, $\bar f = 1/2$, and holds asymptotically when $\mu$ lies below ($\bar f \to 0$) or above ($\bar f \to 1$) the band by many times $T$. For an emitter at the edge of an XY chain, Ref.~\cite{Roos2020} similarly found that divisibility violations of the vacuum dynamics survive at all temperatures. Two consequences follow. First, the bath-state dependence of $\cN_\text{RHP}$ resides entirely in the excess $\cN_\text{RHP} - \cN^{(0)}_\text{RHP}$, which requires $f(\omega)$ to vary within the band; at finite temperature this variation is the Fermi edge, the region of width $\sim T$ around $\omega = \mu$ where $f$ drops from $1$ to $0$. Second, since $\cN^{(0)}_\text{RHP}$ vanishes exactly when $|G(t)|$ is non-increasing, the divisibility measure obeys the analogue of Eq.~\eqref{eq:msvinfT},
\begin{align}
    \cN_\text{RHP}(\beta \to 0) &= \cN^{(0)}_\text{RHP}, \label{eq:rhpinfT} \\
    \cN_\text{RHP}(\beta \to 0) &= 0 \iff \cN_\text{BLP} = 0. \label{eq:rhpinfTzero}
\end{align}

The same identity gives Eq.~\eqref{eq:msvinfT} directly for the chain, where $K \to \infty$ and no normalizable maximally mixed reservoir exists. At $\bar f = 1/2$ the map has no shift, and its intermediate maps $\cM_t\cM_{t'}^{-1}$, $t \geq t'$, simply contract the Bloch ball, up to a rotation about $z$, by $\lambda = |G(t)|/|G(t')|$ in the $xy$ plane and by $\lambda^2$ along $z$, which is completely positive if and only if $\lambda \leq 1$~\cite{Fujiwara1999}. If $|G(t)|$ is non-increasing the dynamics is therefore CP-divisible, and since the QJSD cannot increase under completely positive trace-preserving maps~\cite{Lindblad1975}, no pair revives and $\cN_\text{MSV}(\beta \to 0) = 0$; if $|G|$ rises, the polar pair revives with it.

\section{Numerical Results} \label{sec:results}

We now evaluate the five measures for the spin-chain reservoir at three parameter points, one in each regime of the propagator set by the bound-state threshold $|E| = \Omega - \Gamma/2$ of Sec.~\ref{sec:model}: above it, where a localized level keeps $|G(t)|$ finite (Fig.~\ref{fig:E_2.0_Gamma_0.5}); on it, where $|G(t)|$ decays monotonically (Fig.~\ref{fig:E_0.5_Gamma_1.0}); and below it, where $|G(t)|$ decays to zero through the band (Fig.~\ref{fig:E_0.5_Gamma_0.5}). Figure~\ref{fig:blind} locates the three points in the $(E, \Gamma)$ plane, and Fig.~\ref{fig:mu} varies the chemical potential at fixed temperature.

\paragraph*{Bath-state independence.} Figure~\ref{fig:E_2.0_Gamma_0.5} shows the dynamics and the five measures at $E = 2\Omega$, $\Gamma = \Omega/2$, above the threshold, where a localized level exists. The propagator stays close to the residue $R_b = 0.965$, with $\min_t|G(t)| = 0.941$, so by Appendix~\ref{app:BLP} the polar pair is optimal in every rising window and Eq.~\eqref{eq:blpbound} holds with equality. Across all $26$ temperatures computed, $T = 0$ and $25$ values from $0.01 \Omega$ to $10 \Omega$, $\cN_\text{BLP} = 0.0918$ and $\cN_\text{LPP} = 0.1692$ are identical to machine precision, as Theorem~\ref{thm:independence} requires, even though $Z_\beta(t)$, and with it the rates $\Lambda^\pm(t)$, change with temperature. The time-local master equation~\eqref{eq:master_equation}, which carries the temperature explicitly through $\Lambda^\pm(t)$, returns the same values (open circles). Figure~\ref{fig:E_2.0_Gamma_0.5}(b) starts at $T = 0.01 \Omega$, since at $T = 0$ the three bath-sensitive measures grow logarithmically with $t_\text{max}$ (see Appendix~\ref{app:NM}); at every $T > 0$ the residual truncation at $\Omega t_\text{max} = 320$ is $\approx6\%$ for every measure and temperature alike, and so drops out of the temperature dependence.

\begin{figure*}[t]
    \centering
    \includegraphics[width = \linewidth]{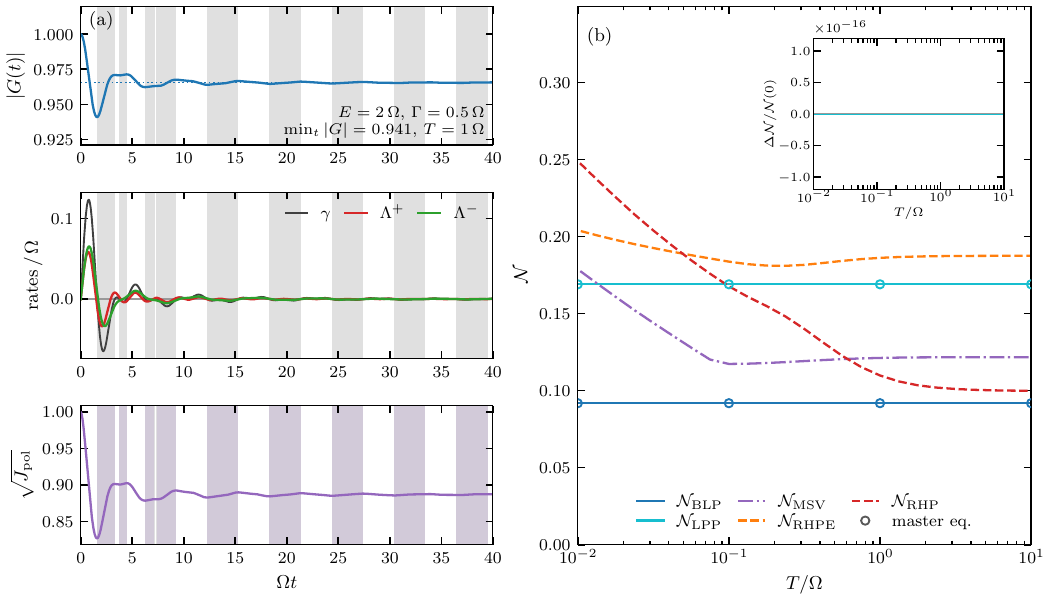}
    \caption{
    (a) Propagator and canonical rates for $E = 2\Omega$, $\Gamma = \Omega/2$ at $T = \Omega$, shown to $\Omega t = 40$; the measures in (b) are integrated over the full window $\Omega t_\text{max} = 320$ ($n_t = 32001$). Shading marks the intervals with $\gamma(t) < 0$, on which $|G(t)|$ rises and backflow accumulates, and the dotted line is the bound-state residue $R_b$; here $\min_t|G(t)| = 0.941$. (b) The five measures of Sec.~\ref{sec:NM} against reservoir temperature at the same parameters. $\cN_\text{BLP}$ and $\cN_\text{LPP}$ are independent of $T$, as Theorem~\ref{thm:independence} requires, while the other three vary by $11$--$60\%$. Inset: deviation of $\cN_\text{BLP}$ and $\cN_\text{LPP}$ from their lowest-temperature values; note the $10^{-16}$ scale, i.e.\ they are constant to machine precision. Open circles are an independent evaluation through the time-local master equation~\eqref{eq:master_equation}, which carries the $T$-dependent rates $\Lambda^\pm(t)$ explicitly, agreeing with the direct evaluation to $4.3 \times 10^{-6}$ (BLP) and $1.1 \times 10^{-4}$ (LPP).
}
 \label{fig:E_2.0_Gamma_0.5}
\end{figure*}

\paragraph*{Entropic measure.} $\cN_\text{MSV}$, evaluated as the best of the families of Appendix~\ref{app:MSV}, Eq.~\eqref{eq:msvbound}, and hence a lower bound within a few percent of an unrestricted optimization (Table~\ref{tab:msvfamilies}), varies by $34\%$ over the same range. It falls steeply at low temperature, from $0.179$ at $T = 0.01 \Omega$ to a minimum of $0.117$ at $T \simeq 0.1 \Omega$, then rises slowly to a plateau, $0.122$ at $T = 10\Omega$. The plateau is nonzero because $\cN_\text{BLP} > 0$ here: by Eq.~\eqref{eq:msvinfT}, $\cN_\text{MSV}$ survives the unital limit whenever the trace distance itself revives. The minimum is not a smooth feature but a change of the optimal pair of input states. Resolving $\cN_\text{MSV}$ into the families of Appendix~\ref{app:MSV}, the best pair at low temperature is a diagonal pair of mixed states close to the ground state, $(c_0,\Delta_0) \simeq (0.11,0.22)$ at $T = 0.02\Omega$ in the notation of Eq.~\eqref{eq:Jpop}, whose advantage over the polar pair derives from the nearly empty bath and shrinks as $Z_\beta$ grows. At a crossover temperature $0.075 \Omega < T^* < 0.1 \Omega$ that advantage vanishes and the polar pair takes over, after which $\cN_\text{MSV}$ follows the nearly temperature-independent antipodal branch. The crossover is the finite-temperature counterpart of the family-(a) enhancement in Table~\ref{tab:msvfamilies}, and it is itself a bath-state effect: the pair that best measures backflow moves with temperature, whereas the optimum of $\cN_\text{BLP}$ does not. Its location does not depend on $t_\text{max}$, and the depth of the minimum below the plateau, $3.6\%$, changes by less than $0.1$ percentage points between $\Omega t_\text{max} = 160$ and $320$.

\paragraph*{Entanglement and divisibility measures.} $\cN_\text{RHP}$ decreases monotonically with temperature, from $0.249$ to $0.0999$, the largest variation of the five measures ($60\%$). Its high-temperature value is the bath-independent floor of Eq.~\eqref{eq:rhpfloor}: evaluated on the same window, $\cN^{(0)}_\text{RHP} = 0.0998$, which $\cN_\text{RHP}$ approaches from above, as Eq.~\eqref{eq:rhpinfT} requires. The excess over the floor accumulates on intervals where one canonical rate is negative while the other is positive; as $T \to \infty$ both tend to $\gamma/2$, they change sign together, and the excess disappears. $\cN_\text{RHPE}$ varies least, by $11\%$, and not monotonically: it falls from $0.204$ to a shallow minimum of $0.181$ at $T \simeq 0.24 \Omega$ before recovering to $0.188$. The concurrence~\eqref{eq:concurrence} subtracts $\sqrt{Z_\beta(1 - |G|^2 - Z_\beta)}$, which is largest at $Z_\beta = (1 - |G|^2)/2$, i.e.\ in the unital limit, so heating degrades the entanglement itself; $\cN_\text{RHPE}$, however, counts the revivals of the concurrence rather than its size, and the two need not move together. In particular, $\cN_\text{RHPE}$ and $\cN_\text{RHP}$ move in opposite directions only above $T \simeq 0.24 \Omega$; below it both decrease.

\paragraph*{On the threshold.} Figure~\ref{fig:E_0.5_Gamma_1.0} shows $E = \Omega/2$, $\Gamma = \Omega$, which lies on the bound-state threshold $|E| = \Omega - \Gamma/2$, at the end of the segment on which $\cN_\text{BLP}$ vanishes (Fig.~\ref{fig:blind}); here the measures are converged within $\Omega t_\text{max} = 40$ at every $T > 0$. This is the Markovianity point $h/J = 1/2$, $J_0/J = 1$ of Ref.~\cite{Apollaro2011} and, in the notation of Ref.~\cite{Roos2020}, the detuning $\Delta_h = J$ at coupling $J$, which they identified as BLP-Markovian but not divisible for a partially filled ground state of the chain. Here $\gamma(t) \geq 0$ at all times: $|G(t)|$ decays monotonically, as $t^{-1/2}$, and $\cN_\text{BLP} = \cN_\text{LPP} = 0$ identically at every temperature. Nevertheless, $\cN_\text{RHP} > 0$: $\Lambda^+(t)$ first turns negative near $\Omega t \simeq 3.1$ while $\Lambda^+ + \Lambda^- = \gamma$ stays positive, giving $\cN_\text{RHP} = 0.093$ at $T = 0.02 \Omega$ although $\cN_\text{BLP}$ vanishes. The equal-weight trace distance cannot register a violation that lives in the bath-dependent translation (Sec.~\ref{sec:qjsd}). The entropic measure does: $\sqrt{J_\text{pol}}$ revives on the shaded intervals of Fig.~\ref{fig:E_0.5_Gamma_1.0}(a), and $\cN_\text{MSV} = 9.4 \times 10^{-3}$ at $T = 0.02 \Omega$. Both measures vanish in the unital limit, as Eqs.~\eqref{eq:msvinfT} and~\eqref{eq:rhpinfTzero} require when $\cN_\text{BLP} = 0$: $\cN_\text{MSV}$ falls to $1.8 \times 10^{-6}$ at $T = 10 \Omega$, and $\cN_\text{RHP}$ to $3.7 \times 10^{-4}$. $\cN_\text{RHPE}$ behaves differently again, vanishing identically for $T \gtrsim 0.5 \Omega$. The concurrence revives once, on $\Omega t \simeq 3.3$--$4.4$, and then suffers sudden death; heating brings the death forward, from $\Omega t = 5.5$ at $T = 0$ to $3.3$ at $T = 0.58 \Omega$, until it precedes the revival and removes it. Since $|G(t)|$ is not rising during that window, the revival is driven entirely by $Z_\beta(t)$, another manifestation of the non-unital shift that the trace distance cannot see.

\begin{figure*}[t]
    \centering
    \includegraphics[width = \linewidth]{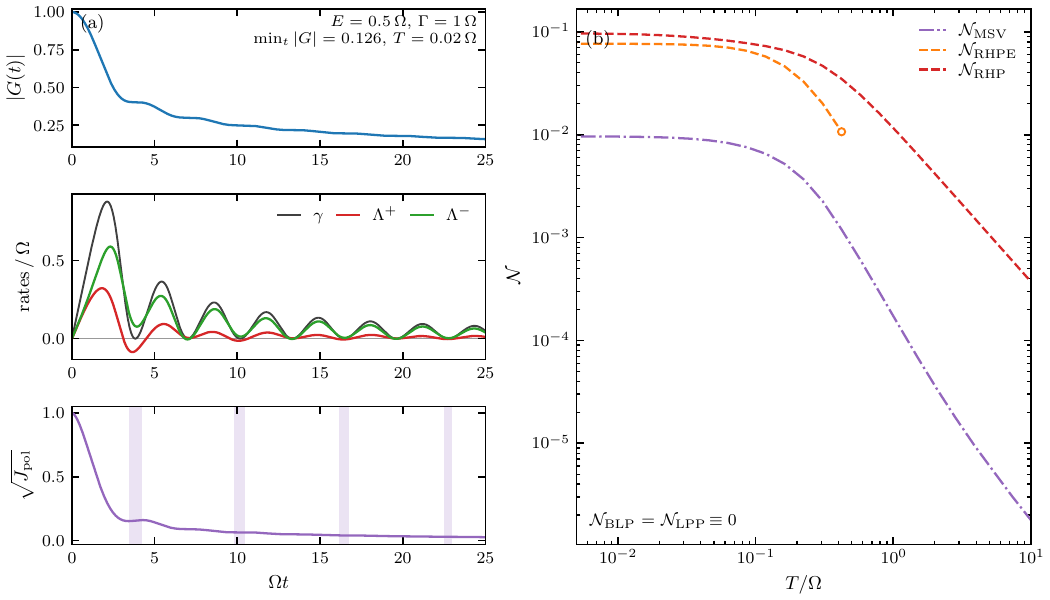}
    \caption{
    As Fig.~\ref{fig:E_2.0_Gamma_0.5}, for $E = \Omega/2$, $\Gamma = \Omega$ with panel (a) at $T = 0.02 \Omega$. This point lies on the bound-state threshold $|E| = \Omega - \Gamma/2$, where $|G(t)| \sim t^{-1/2}$ and $\min_t|G(t)| = 0.126$. Panel (a) is shown to $\Omega t = 25$; the measures in (b) are integrated to $\Omega t_\text{max} = 40$, at which they are converged here for every $T > 0$ (see Appendix~\ref{app:NM}). (a) No interval has $\gamma(t) < 0$, so $|G|$ decreases monotonically, no shading appears in the upper two panels, and $\cN_\text{BLP} = \cN_\text{LPP} = 0$ identically. The individual rate $\Lambda^{+}(t)$ nevertheless first turns negative near $\Omega t \simeq 3.1$, and $\sqrt{J_\text{pol}}$ revives on the shaded intervals: the \emph{sum} of the canonical rates stays positive while one of them does not. (b) Measures against temperature, both axes logarithmic; $\cN_\text{BLP}$ and $\cN_\text{LPP}$ vanish identically and are omitted from the legend. $\cN_\text{MSV}$ is positive at every finite temperature, decreasing monotonically and vanishing only in the unital limit (see Eq.~\eqref{eq:msvinfT}). $\cN_\text{RHPE}$ instead reaches zero at $T \simeq 0.5 \Omega$ and stays there (the open circle marks the last nonzero point), because by $T \simeq 0.58 \Omega$ the concurrence suffers sudden death at $\Omega t \simeq 3.3$, right at the start of its revival window $\Omega t \simeq 3.3\text{--}4.4$, which is then cut off entirely.
}
 \label{fig:E_0.5_Gamma_1.0}
\end{figure*}

\paragraph*{Below the threshold.} Figure~\ref{fig:E_0.5_Gamma_0.5} shows $E = \Gamma = \Omega/2$: the coupling of Fig.~\ref{fig:E_2.0_Gamma_0.5}, but below the threshold, which lies at $E_c = 0.75 \Omega$ for this $\Gamma$, and inside the regime $\sqrt{E^2 + \Gamma\Omega} < \Omega$ of Ref.~\cite{Mitchison2018}. No localized level exists, and $|G(t)|$ decays through the band, falling below $10^{-2}$ by $\Omega t \simeq 20$. From then on $G(t)$ is dominated by the two band edges, where $\varphi(\omega)$ vanishes as a square root (Appendix~\ref{app:longtime}): each edge contributes a term decaying as $t^{-3/2}$ and oscillating at its own frequency, $\pm\Omega$, and the beating between the two makes $\gamma(t)$ turn negative once in every period $\pi/\Omega$ for all $\Omega t \gtrsim 20$ (shaded in Fig.~\ref{fig:E_0.5_Gamma_0.5}(a)), so that $|G(t)|$ revives, weakly but indefinitely. The measures respond to these revivals in two distinct ways. $\cN_\text{BLP}$, $\cN_\text{LPP}$ and $\cN_\text{MSV}$ are built from the distinguishability of states whose Bloch vectors have shrunk with $|G|$, so the late revivals add to them only slowly. $\cN_\text{BLP} = 4.9 \times 10^{-3}$ and $\cN_\text{LPP} = 1.8 \times 10^{-8}$ are once more identical at every temperature, and $\cN_\text{MSV}$ falls from $4.9 \times 10^{-3}$ to $4.2 \times 10^{-3}$; at $\Omega t_\text{max} = 1280$, $\cN_\text{BLP}$ reaches $1.26 \times 10^{-2}$ and converges as $t_\text{max}^{-1/2}$. The canonical rates carry no such weight. Once $|G|$ is small, $Z_\beta(t)$ stops changing, so both rates follow $\gamma(t)$ and change sign with it, as for a flat $f(\omega)$. The dynamics therefore never becomes CP-divisible; this non-Markovianity is recurrent rather than eternal in the sense of Ref.~\cite{Hall2014}, since the rates are negative only on intervals. $\cN_\text{RHP}$ accordingly accumulates along the floor of Eq.~\eqref{eq:rhpfloor}, which grows without bound with the window: $\cN^{(0)}_\text{RHP} = 1.06$, $2.16$ and $5.27$ at $\Omega t_\text{max} = 40$, $80$ and $160$. Below the threshold $\cN_\text{RHP}$ therefore has no limit as $t_\text{max} \to \infty$, and its dependence on the bath state is confined to a finite excess over the floor, $4.4 \times 10^{-3}$ at $T = 0.02 \Omega$, from a brief negative excursion of $\Lambda^+$ near $\Omega t \simeq 4$ that disappears at high temperature. $\cN_\text{RHPE}$ vanishes identically: the concurrence suffers sudden death by $\Omega t \simeq 7.7$ at $T = 0.02 \Omega$, earlier when hotter, and the revivals of $|G|$ are too weak to restore it. With $\cN_\text{RHP}$ nearly flat in temperature, $\cN_\text{RHPE}$ vanishing and $\cN_\text{MSV}$ varying by less than half as much as above the threshold, this regime complements Fig.~\ref{fig:E_0.5_Gamma_1.0}: a strong dependence on the bath state requires $|G(t)|$ to remain appreciable while $Z_\beta(t)$ evolves.

\begin{figure*}[t]
    \centering
    \includegraphics[width = \linewidth]{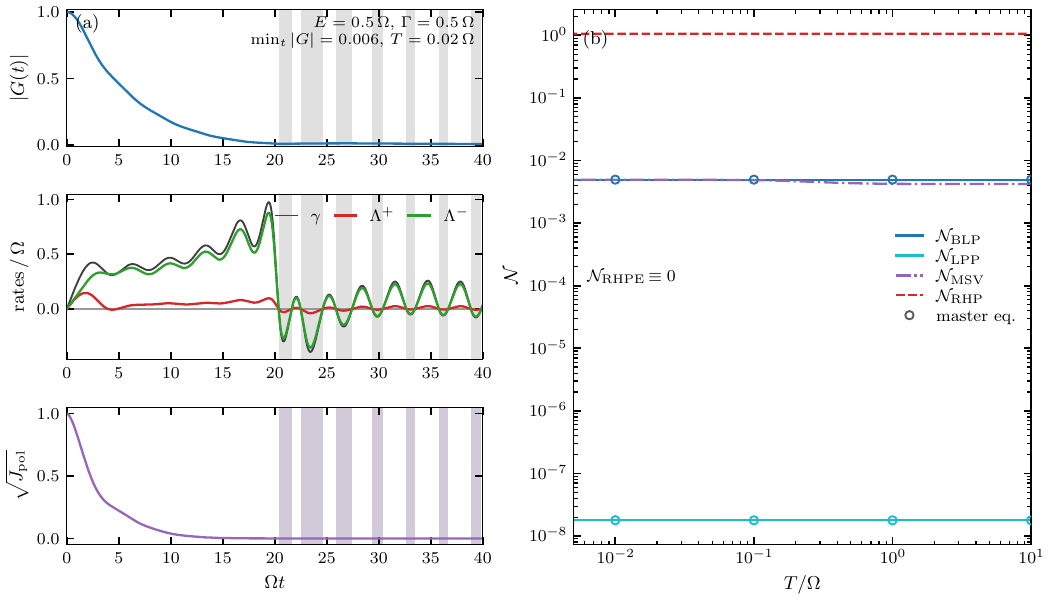}
    \caption{
    As Fig.~\ref{fig:E_2.0_Gamma_0.5}, for $E = \Gamma = \Omega/2$, below the bound-state threshold ($E_c = 0.75 \Omega$ for this $\Gamma$), with panel (a) at $T = 0.02 \Omega$ and the measures in (b) integrated to $\Omega t_\text{max} = 40$. (a) $|G(t)|$ decays through the band, with $\min_t|G(t)| = 0.006$. From $\Omega t \simeq 20$ the two band edges beat at frequency $2\Omega$: $\gamma(t)$ and both rates $\Lambda^\pm(t)$ turn negative once per period (gray shading), and $\sqrt{J_\text{pol}}$ revives weakly on the shaded intervals of the lowest panel. (b) Measures against temperature, both axes logarithmic. $\cN_\text{BLP}$ and $\cN_\text{LPP}$ are independent of $T$, and the master equation (open circles) reproduces them to $1.0 \times 10^{-6}$ and $7 \times 10^{-11}$, respectively; $\cN_\text{MSV}$ varies by $15\%$ and $\cN_\text{RHP}$ by $0.5\%$, and $\cN_\text{RHPE}$ vanishes identically. $\cN_\text{RHP}$ lies on the bath-independent floor of Eq.~\eqref{eq:rhpfloor}, which grows without bound with $t_\text{max}$ in this regime, so its value is specific to the window.
    }
    \label{fig:E_0.5_Gamma_0.5}
\end{figure*}

\paragraph*{Where the trace distance is blind.} Since $\cN_\text{BLP}$ and $\cN_\text{LPP}$ depend on $|G(t)|$ alone, whether they vanish is a property of $(E, \Gamma)$ and holds for every bath state at once. Figure~\ref{fig:blind} maps $\cN_\text{BLP}$ over the $(E, \Gamma)$ plane on a window $\Omega t_\text{max} = 1280$, with the parameter points of Figs.~\ref{fig:E_2.0_Gamma_0.5}--\ref{fig:E_0.5_Gamma_0.5} marked. It vanishes identically only on the bound-state threshold, and only for $\Gamma \leq \Omega$; the segment ends at $\Gamma = \Omega$, which is the point of Fig.~\ref{fig:E_0.5_Gamma_1.0}. This segment coincides with the BLP-Markovian points found for the same model in Refs.~\cite{Apollaro2011, Roos2020}: with hopping $J = \Omega/2$ and qubit coupling $J_0 = \Omega_\text{R} = \sqrt{\Gamma\Omega}/2$, they lie at detuning $2h = \Delta_h = 2J - J_0^2/J$ for $J_0 \le J$. Everywhere else $|G(t)|$ eventually rises: above the threshold through the interference of the localized level with the band, below it through the band-edge beating of Fig.~\ref{fig:E_0.5_Gamma_0.5}. Close to the threshold the first revival comes ever later, so the thin gray band along the threshold in Fig.~\ref{fig:blind} narrows as $t_\text{max}$ grows. Along $\Gamma = \Omega/2$ it occurs at $\Omega t \simeq 61$, $152$ and $633$ for $E - E_c = -0.15$, $-0.10$ and $-0.05 \Omega$, in agreement with Eq.~\eqref{eq:tfirst}, which gives $61$, $147$ and $631$. The bath-sensitive measures have no such zero set: at the point of Fig.~\ref{fig:E_0.5_Gamma_1.0}, where $\cN_\text{BLP} = 0$, $\cN_\text{RHP}$ and $\cN_\text{MSV}$ remain finite at every finite temperature.

\begin{figure}[t]
    \centering
    \includegraphics[width = \linewidth]{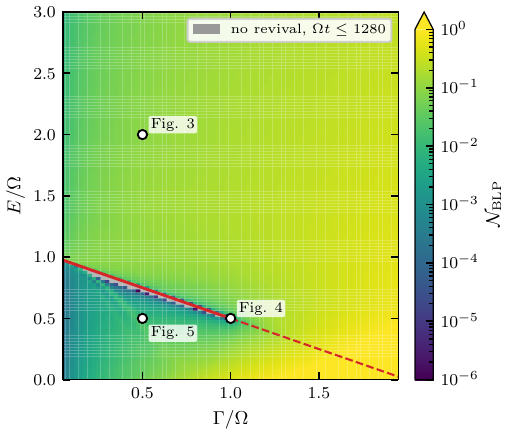}
    \caption{
    $\cN_\text{BLP}$ over the $(E, \Gamma)$ plane, evaluated on $\Omega t_\text{max} = 1280$; since it depends on $|G(t)|$ alone, the map holds for every state of the reservoir, and $\cN_\text{LPP}$ vanishes on the same set. Gray: no revival of $|G(t)|$ within the window. Red line: the bound-state threshold $|E| = \Omega - \Gamma/2$, solid for $\Gamma \leq \Omega$, where $\cN_\text{BLP} = 0$ exactly, and dashed beyond. The gray band along the solid segment narrows as $t_\text{max}$ grows, since revivals near the threshold set in at times $\propto(E - E_c)^{-2}$, Eq.~\eqref{eq:tfirst}. Circles mark the parameters of Fig.~\ref{fig:E_2.0_Gamma_0.5} (above the threshold), Fig.~\ref{fig:E_0.5_Gamma_1.0} (on it, at the end of the blind segment) and Fig.~\ref{fig:E_0.5_Gamma_0.5} (below it). Only $E \geq 0$ is shown, since $|G(t)|$ is invariant under $E \to -E$; the color scale is clipped at $1$.
    }
    \label{fig:blind}
\end{figure}

\paragraph*{Chemical potential.} Theorem~\ref{thm:independence} concerns the bath \emph{state} as a whole, of which temperature is only one parameter, with the chemical potential being another. By Eq.~\eqref{eq:Zbound} the noise term is confined to $0 \leq Z_\beta(t) \leq 1 - |G(t)|^2$. At $\mu = 0$ a temperature scan explores part of this interval; a scan of $\mu$ through the band at low temperature explores all of it, from an empty band ($\mu < -\Omega$) to a filled one ($\mu > \Omega$). Figure~\ref{fig:mu} shows the result at $T = 0.05 \Omega$ for the three parameter points of Figs.~\ref{fig:E_2.0_Gamma_0.5}--\ref{fig:E_0.5_Gamma_0.5}. Although the filling $Z_\beta/(1 - |G|^2)$ runs over its full range $[0,1]$, $\cN_\text{BLP}$ and $\cN_\text{LPP}$ are unchanged to machine precision across all $61$ values of $\mu$. The bath-sensitive measures, by contrast, change by factors of $1.6$ ($\cN_\text{MSV}$), $2.0$ ($\cN_\text{RHP}$) and $2.4$ ($\cN_\text{RHPE}$) at $E = 2\Omega$, and by several orders of magnitude on the threshold, at $E = \Omega/2$, $\Gamma = \Omega$; below it $\cN_\text{MSV}$ still changes by a factor of $2.1$, while $\cN_\text{RHP}$, dominated by its bath-independent floor, changes by only $2\%$.

\begin{figure*}[t]
    \centering
    \includegraphics[width = \linewidth]{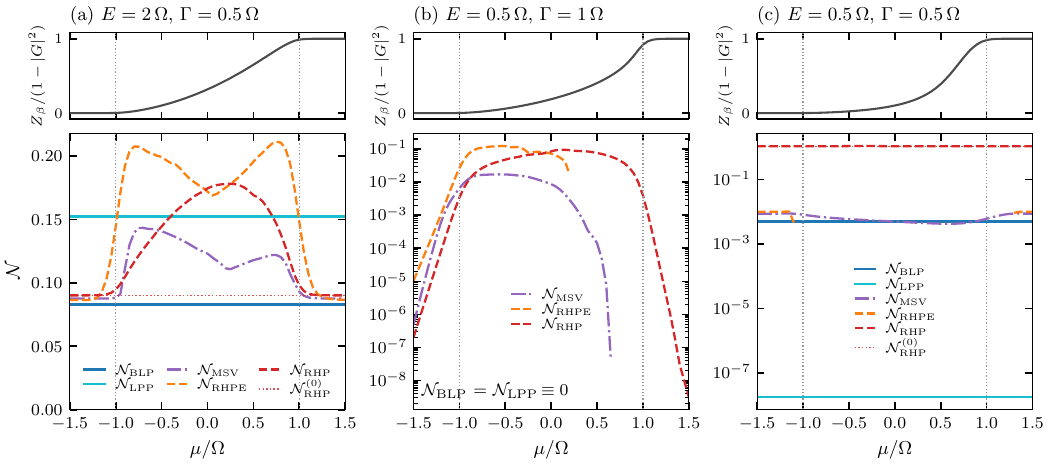}
    \caption{
    Measures against the chemical potential of the reservoir at $T = 0.05 \Omega$, for (a) $E = 2\Omega$, $\Gamma = \Omega/2$, (b) $E = \Omega/2$, $\Gamma = \Omega$ and (c) $E = \Gamma = \Omega/2$ --- above, on and below the bound-state threshold, the points of Figs.~\ref{fig:E_2.0_Gamma_0.5}--\ref{fig:E_0.5_Gamma_0.5} --- with (b) and (c) on a logarithmic scale. Top strips: the filling $Z_\beta(t_\text{max})/[1 - |G(t_\text{max})|^2]$, which spans its full range $[0,1]$ as $\mu$ crosses the band (dotted lines, $\mu = \pm\Omega$). $\cN_\text{BLP}$ and $\cN_\text{LPP}$ are unchanged to machine precision throughout, and in (b) vanish identically. Outside the band, where $f(\omega)$ is flat over the spectral support, the bath-sensitive measures reach plateaus that agree between the empty and the filled bath, and $\cN_\text{RHP}$ meets the bath-independent floor $\cN^{(0)}_\text{RHP}$ of Eq.~\eqref{eq:rhpfloor} (dotted). In (b) that floor is zero and all bath-sensitive non-Markovianity is confined to the band; in (c) it dominates, and $\cN_\text{RHP}$ exceeds it by at most $2\%$. Window $\Omega t_\text{max} = 40$: the flatness of $\cN_\text{BLP}$ and $\cN_\text{LPP}$, and the comparison with the floor, which is evaluated on the same window, hold at any $t_\text{max}$.
    }
    \label{fig:mu}
\end{figure*}

The shape of these curves is set by where the Fermi edge lies relative to the band, through the flat-$f$ argument below Eq.~\eqref{eq:rhpfloor}. Where $\mu$ lies outside the band, $f(\omega)$ is flat over the spectral support and the bath-sensitive measures settle onto plateaus. The empty and filled bands give the same values, since their maps differ only by $\bar f \to 1 - \bar f$, a particle--hole conjugation that leaves every measure invariant: at $E = 2\Omega$, $\cN_\text{MSV} = 0.0877$ at both ends, equal to within $5 \times 10^{-7}$. There $\cN_\text{RHP}$ also coincides with the floor $\cN^{(0)}_\text{RHP} = 0.0901$ computed from $G(t)$ alone. On the threshold, at $E = \Omega/2$, $\Gamma = \Omega$, the floor vanishes, because $\gamma(t) \geq 0$, and all three bath-sensitive measures are confined to the window in which the Fermi edge lies inside the band: $\cN_\text{RHP}$ falls from $0.095$ to $3 \times 10^{-9}$ as $\mu$ leaves the band from above. The non-Markovianity these measures register at this point is produced entirely by the Fermi edge. Below the threshold, at $E = \Gamma = \Omega/2$, the floor dominates instead: $\cN_\text{RHP}$ equals $\cN^{(0)}_\text{RHP} = 1.057$ outside the band and exceeds it by at most $0.021$ inside, while $\cN_\text{MSV}$ varies between $4.2 \times 10^{-3}$ and $8.7 \times 10^{-3}$, with equal plateaus at both ends. $\cN_\text{RHPE}$ is nonzero only where the band is nearly empty or nearly full, $\mu \lesssim - 1.1 \Omega$ or $\mu \gtrsim 1.3\Omega$. There $Z_\beta \simeq \bar f (1 - |G|^2)$ with $\bar f \to 0$ or $1$, the concurrence~\eqref{eq:concurrence} reduces to $|G(t)|$, and $\cN_\text{RHPE} = 9.9 \times 10^{-3}$ is twice the trace-distance value $\cN_\text{BLP} = 4.94 \times 10^{-3}$, since it counts each revival of $|G|$ once on the way up and once on the way down; at intermediate filling the concurrence dies before $|G|$ revives. At $T = 0.5 \Omega$ the same picture holds with a smoother dependence on $\mu$, and $\cN_\text{BLP}$ and $\cN_\text{LPP}$ again do not move.

\section{Conclusion} \label{sec:conclusion}

Our work explores the connection between non-Markovianity measures and the parameters of the environment in fermionic quadratic open quantum systems. As our first main result, we have proven that the Breuer--Laine--Piilo and Lorenzo--Plastina--Paternostro measures, built on the trace distance and on the volume of accessible states, respectively, are independent of the state of the reservoir, and in particular of its temperature and chemical potential (Theorem~\ref{thm:independence}). The proof rests on the affine structure of the reduced dynamics: for a single fermionic site coupled to an arbitrary quadratic reservoir, pairing terms included, the bath state enters only through a single scalar, $Z_\beta(t)$, multiplying a traceless shift along $\hat\sigma^z$, Eq.~\eqref{eq:affine}, which cancels from both measures. The only assumptions are an initial product state between the system and the environment and a parity-even reservoir, $[\hat\rho_B,\hat\cP] = 0$, with $\hat\cP$ defined in Eq.~\eqref{eq:parity}. A bath-state-dependent BLP or LPP measure therefore requires either interactions beyond quadratic order or initial system--bath correlations.

By contrast, the entanglement- and divisibility-based Rivas--Huelga--Plenio measures and the entropic Megier--Smirne--Vacchini measure are sensitive to the reservoir state. For the quantum Jensen--Shannon divergence this happens through the mean occupation of the pair of input states, which the trace distance ignores (Theorem~\ref{thm:qjsd}); in the unital limit the dependence disappears and $\cN_\text{MSV}$ vanishes exactly when $\cN_\text{BLP}$ does, Eq.~\eqref{eq:msvinfT}. For the divisibility measure we obtained a bath-independent floor fixed by the propagator alone, Eq.~\eqref{eq:rhpfloor}, saturated whenever the occupation of the reservoir is flat over the band, so that all of its dependence on temperature and chemical potential is carried by the Fermi edge within the band, Eq.~\eqref{eq:rhpinfT}.

To model the dynamics of a single fermionic site coupled to a free-fermionic chain, we use an exact time-local master equation, Eq.~\eqref{eq:master_equation}, valid both when the system energy lies within the band and when a localized level splits off from it. It takes the form of the master equation of Ref.~\cite{Mitchison2018} for a pair of coupled fermionic modes in contact with two reservoirs, of which our model is a special case.

Numerical simulations for this model illustrate the general results at three points, above, on and below the bound-state threshold. Above it, $\cN_\text{BLP}$ and $\cN_\text{LPP}$ are constant to machine precision over three orders of magnitude in temperature and across the full range of band filling, and the time-local master equation, which carries the bath state explicitly through its rates, reproduces them independently; over the same temperatures $\cN_\text{MSV}$, $\cN_\text{RHPE}$ and $\cN_\text{RHP}$ vary by $34\%$, $11\%$ and $60\%$, respectively, and not always monotonically. The minimum of $\cN_\text{MSV}$ marks a change of the optimal pair of input states, itself a bath-state effect. On the threshold, at $E = \Omega/2$, $\Gamma = \Omega$, $\cN_\text{BLP} = \cN_\text{LPP} = 0$ identically, while the three bath-sensitive measures remain finite at low temperature: there the non-Markovianity lives entirely in the bath-dependent translation of the Bloch ball, which the equal-weight trace distance cannot register, although the Helstrom norm with unequal priors does. Below the threshold the propagator decays through the band and the dependence on the bath state becomes much weaker than above it, while the beating of the two band edges makes $\cN_\text{RHP}$ grow without bound with the integration window.

Theorem~\ref{thm:independence} and the map~\eqref{eq:mapC} hold for arbitrary fermionic quadratic Hamiltonians, pairing terms included, but the explicit forms of the measures obtained here rely on $F = 0$: with $F \neq 0$ the BLP optimization no longer reduces to polar and equatorial pairs (Appendix~\ref{app:BLP}), and only $\cN_\text{LPP}$ carries over unchanged. A natural next step is therefore an explicit time-local master equation for the general fermionic quadratic model. We also expect the BLP and LPP measures to depend on the bath state once quartic interactions are present, since these allow the bath state to enter the linear part of the reduced map, and not only the shift; when, and how strongly, this dependence appears in specific interacting models remains an open question. Finally, we note that the bath-state independence of $\cN_\text{BLP}$ does not extend to quadratic bosonic models, where a temperature dependence has been reported for the Caldeira--Leggett model~\cite{Einsiedler2020, Manara2026, Pellitteri2026}; understanding the origin of this difference between fermionic and bosonic quadratic systems is an interesting direction for future work.

\section*{Acknowledgments}

We thank Mark T. Mitchison, and T. J. G. A. thanks Bassano Vacchini, for fruitful discussions. M. C. acknowledges funding from Project QAMALA (Quantum Algorithms and MAchine LeArning) financed by the Maltese Ministry for Education, Sport, Youth, Research, and Innovation (MEYR) Grant ``2025301 UM MinED''. I. K. and T. J. G. A. acknowledge co-funding by the European Union (Quantum Flagship project ASPECTS, Grant Agreement No.\ 101080167). Views and opinions expressed are however those of the authors only and do not necessarily reflect those of the European Union or the Research Executive Agency. Neither the European Union nor the granting authority can be held responsible for them.

\section*{Data Availability}

The measures were computed directly from $G(t)$ and $Z_\beta(t)$ with our own Python code based on NumPy and SciPy; the master-equation cross-checks were performed with QuTiP~\cite{Lambert2026}. The code and data supporting the figures are available at~\url{https://github.com/aspects-quantum/bath-state-nonmarkovianity}.

\section*{Artificial Intelligence Disclosure}

Claude (Opus 5.5 and 5.0) was used, under the authors' direction, to assist with literature searches, the examination of mathematical arguments, development of the codebase, and the preparation and revision of the manuscript. The authors checked all scientific statements, derivations, and references, and take full responsibility for the content of this work.

\onecolumngrid
\appendix

\section{Details for Solving the Dynamical Map with \texorpdfstring{$F = 0$}{F = 0}} \label{app:dynamical_map}

Since the Hamiltonian~\eqref{eq:H} is quadratic, the Heisenberg equations of motion are linear. In the eigenmode basis~\eqref{eq:aq} they read
\begin{align}
    \ii\partial_t\hat d(t) &= E\hat d(t) + \sum_q \tau_q\hat a_q(t), \label{eq:eom_d} \\
    \ii\partial_t\hat a_q(t) &= \omega_q\hat a_q(t) + \tau_q^*\hat d(t). \label{eq:eom_a_q}
\end{align}
Following Ref.~\cite{Mitchison2018}, we solve these in the Laplace domain, writing $\tilde A(z) = \int_0^\infty e^{-zt} \hat A(t) \dd{t}$. Transforming Eqs.~\eqref{eq:eom_d} and~\eqref{eq:eom_a_q} turns them into an algebraic system,
\begin{align}
 (z + \ii E)\tilde d(z) + \ii\sum_q\tau_q\tilde a_q(z) &= \hat d, \\
 (z + \ii\omega_q)\tilde a_q(z) + \ii\tau_q^*\tilde d(z) &= \hat a_q,
\end{align}
whose solution for the system row is obtained by eliminating $\tilde a_q$,
\begin{equation}
    \tilde d(z) = \tilde G(z)\hat d - \ii\tilde G(z)\sum_q\frac{\tau_q}{z + \ii\omega_q}\hat a_q,~\tilde G(z) = \left[z + \ii E + \tilde W(z)\right]^{-1},
    \label{eq:laplace-sol}
\end{equation}
with the reservoir kernel
\begin{equation}
    \tilde W(z) = \sum_q\frac{|\tau_q|^2}{z + \ii\omega_q} = \int \frac{\cJ(\omega)}{z + \ii\omega} \dd{\omega}.
    \label{eq:Wtilde}
\end{equation}
Inverting Eq.~\eqref{eq:laplace-sol} term by term gives
\begin{equation}
    \hat d(t) = G(t)\hat d + \hat\xi(t),~\hat\xi(t) = \sum_q g_q(t)\hat a_q,
    \label{eq:propagatorA}
\end{equation}
with $G(0) = 1$ and $g_q(0) = 0$, now in the eigenmode basis. The noise amplitudes follow immediately from the convolution theorem, since Eq.~\eqref{eq:laplace-sol} is a product of $\tilde G(z)$ with
$(z + \ii\omega_q)^{-1}$:
\begin{equation}
    g_q(t) = -\ii\tau_q\int_0^t G(t - s)e^{-\ii\omega_q s} \dd{s} = -\ii\tau_q e^{-\ii\omega_q t}g(\omega_q,t),~g(\omega,t) \equiv \int_0^t e^{\ii\omega t'}G(t') \dd{t'}.
\end{equation}

The kernel~\eqref{eq:Wtilde} is the impurity self-energy in disguise. Continuing $z \to -\ii\omega + 0^+ $ gives $\tilde W(-\ii\omega + 0^+ ) = \ii\Sigma(\omega)$, with $\Sigma(\omega)$ the retarded self-energy of Eq.~\eqref{eq:selfenergy}, so that
\begin{equation}
    \tilde G(-\ii\omega + 0^+ ) = \frac{\ii}{\omega - E - \Sigma(\omega)}.
\end{equation}
The analytic structure of $\tilde G$ therefore reproduces the two contributions to the propagator identified in Eq.~\eqref{eq:G-bound}. Inverting the Laplace transform by deforming the Bromwich contour onto the branch cut of $\sqrt{z^2 + \Omega^2}$ joining $z = \pm\ii\Omega$~\cite{Mitchison2018}, the cut contributes $\int_{-\Omega}^{\Omega} e^{-\ii\omega t}\varphi(\omega) \dd{\omega}$ with $\varphi$ given by Eq.~\eqref{eq:varphi}. In addition, $\tilde G$ has a pole wherever $z + \ii E + \tilde W(z) = 0$, i.e.\ at $z_b = -\ii\omega_b$ with $\omega_b = E + \Sigma(\omega_b)$, whose residue is
\begin{equation}
    \left[1 + \tilde W'(z_b)\right]^{-1} = \left[1 - \Sigma'(\omega_b)\right]^{-1} = R_b,
\end{equation}
contributing $R_be^{-\ii\omega_bt}$, in agreement with Eq.~\eqref{eq:G-bound}. Reference~\cite{Mitchison2018} excludes such poles by assumption, requiring $\sqrt{E^2 + \Gamma\Omega} < \Omega$ so that $G(t) \to 0$ and the system relaxes to a unique stationary state; the localized levels of Eq.~\eqref{eq:G-bound} are precisely the poles that appear when that condition fails.

Finally, Heisenberg evolution is unitary, so $\{\hat d(t),\hat d^\dagger(t)\} = \hat U^\dagger\{\hat d,\hat d^\dagger\}\hat U = 1$ at all times. Expanding with Eq.~\eqref{eq:propagatorA} and using $\{\hat a_q,\hat a_{q'}^\dagger\} = \delta_{qq'}$ together with $\{\hat d,\hat a_q^\dagger\} = 0$ yields the sum rule
\begin{equation}
    |G(t)|^2 + \sum_q|g_q(t)|^2 = 1.
    \label{eq:sumrule}
\end{equation}

The bath state~\eqref{eq:gibbs} is Gaussian and commutes with $\hat N_B$, so $\ev*{\hat a_q} = 0$ and $\ev*{\hat a_q^\dagger \hat a_{q'}} = \delta_{qq'}f(\omega_q)$, with $f(\omega) = [e^{\beta(\omega - \mu)} + 1]^{-1}$. Hence $\ev*{\hat\xi(t)} = 0$, and the noise correlator is
\begin{equation}
    Z_\beta(t) \equiv \ev*{\hat\xi^\dagger(t)\hat\xi(t)} = \sum_{q,q'} g_q^*(t)g_{q'}(t)\ev*{\hat a_q^\dagger \hat a_{q'}} = \sum_q |g_q(t)|^2 f(\omega_q) = \int \cJ_\text{N}(\omega) f(\omega) |g(\omega,t)|^2 \dd{\omega},
    \label{eq:noise_app}
\end{equation}
where the last step uses $|g_q(t)|^2 = |\tau_q|^2|g(\omega_q,t)|^2$ together with the definition~\eqref{eq:Jdef} of the spectral density and the continuum limit~\eqref{eq:newns}. Expanding $|g(\omega,t)|^2$ as a double time integral in Eq.~\eqref{eq:noise_app} and performing the frequency integral first recovers Eq.~\eqref{eq:Zdouble} of the main text,
\begin{equation}
    Z_\beta(t) = \int_0^t \int_0^t G^*(t')\Phi(t' - t'')G(t'') \dd{t''} \dd{t'},
\end{equation}
so the two representations are equivalent; the single-integral form~\eqref{eq:noise_app} is manifestly non-negative and is the one used in the numerics. The temperature and chemical potential of the bath enter Eq.~\eqref{eq:noise_app} only through $f$, and do not appear in $G$ or in $g(\omega,t)$, consistent with the general result of Sec.~\ref{sec:quadratic}, here made explicit for the chain.

Applying the same continuum limit to the sum rule~\eqref{eq:sumrule} and using $0 \leq f(\omega) \leq 1$ bounds the noise term by
\begin{equation}
    0 \leq Z_\beta(t) \leq \int \cJ_\text{N}(\omega)|g(\omega,t)|^2 \dd{\omega} = 1 - |G(t)|^2,
    \label{eq:Zbound}
\end{equation}
the two limits corresponding to an empty ($f \to 0$) and a fully occupied ($f \to 1$) bath. Complete positivity was already established in Sec.~\ref{sec:quadratic} on general grounds; Eq.~\eqref{eq:Zbound} makes it explicit for this model, and provides a sharp numerical check on $Z_\beta$.

Setting $F(t) = 0$ in Eqs.~\eqref{eq:coh} and~\eqref{eq:pop} of the main text gives
\begin{equation}
    \ev*{\hat d(t)} = G(t)\ev*{\hat d},~\ev*{\hat d^\dagger(t)\hat d(t)} = |G(t)|^2\ev*{\hat d^\dagger\hat d} + Z_\beta(t),
    \label{eq:cohpopapp}
\end{equation}
the cross terms vanishing as before.

Reading off Eq.~\eqref{eq:cohpopapp}, the coherence gives
\begin{equation}
    \rho_{10}(t) = G(t)\rho_{10}(0),~\rho_{01}(t) = G^*(t)\rho_{01}(0),
\end{equation}
and the population gives
\begin{equation}
    \rho_{11}(t) = |G(t)|^2\rho_{11}(0) + Z_\beta(t).
    \label{eq:rho11}
\end{equation}
Using $\rho_{00}(0) + \rho_{11}(0) = 1$ to write $Z_\beta(t) = Z_\beta(t)[\rho_{00}(0) + \rho_{11}(0)]$,
and $\rho_{00}(t) = 1 - \rho_{11}(t)$, Eq.~\eqref{eq:rho11} becomes
\begin{equation}
    \rho_{11}(t) = \left(|G(t)|^2 + Z_\beta(t)\right)\rho_{11}(0) + Z_\beta(t)\rho_{00}(0),~\rho_{00}(t) = \left(1 - |G(t)|^2 - Z_\beta(t)\right)\rho_{11}(0) + \left(1 - Z_\beta(t)\right)\rho_{00}(0).
    \label{eq:rhodiag}
\end{equation}
Equation~\eqref{eq:Zbound} ensures that all four coefficients in Eq.~\eqref{eq:rhodiag} are non-negative, as complete positivity requires. Thus,
\begin{equation}
    \cM_t =
    \begin{pmatrix}
        1 - Z_\beta(t) & 0 & 0 & 1 - |G(t)|^2 - Z_\beta(t) \\
        0 & G(t) & 0 & 0 \\
        0 & 0 & G^*(t) & 0 \\
        Z_\beta(t) & 0 & 0 & |G(t)|^2 + Z_\beta(t)
    \end{pmatrix}.
    \label{eq:map}
\end{equation}

\section{Derivation of the Measures} \label{app:NM}

This appendix collects the definitions of the five measures introduced in Sec.~\ref{sec:NM} and derives their explicit forms for the model of Sec.~\ref{sec:model}.

The measures are defined again here for convenience:
\begin{align}
    \cN_\text{BLP} &= \max_{\hat\rho_{1,2}(0)} \int_0^\infty \max\left\{0, \partial_t D(\hat\rho_1(t), \hat\rho_2(t))\right\} \dd{t},~D(\hat\rho_1, \hat\rho_2) = \frac{1}{2}\lVert\hat\rho_1 - \hat\rho_2\rVert_1, \\
    \cN_\text{LPP} &= \int_0^\infty \max\left\{0, \partial_t \det\mathbf{A} \right\} \dd{t},~(\mathbf{A})_{ij} = \frac{\Tr{\hat\sigma_i \cM_t(\hat\sigma_j)}}{2},~i, j \in \{x,y,z\}, \\
    \cN_\text{MSV} &= \max_{\hat\rho_{1,2}(0)}\int_0^\infty \max\left\{0,\partial_t\sqrt{J\left(\hat\rho_1(t),\hat\rho_2(t)\right)}\right\}\dd{t},~J(\hat\rho_1,\hat\rho_2) = \cS\left(\frac{\hat\rho_1 + \hat\rho_2}{2}\right) - \frac{1}{2}\left[\cS(\hat\rho_1) + \cS(\hat\rho_2)\right], \\
    \cN_\text{RHPE} &= \int_0^\infty \left| \partial_t \cE(\hat\rho_{SA}(t)) \right| \dd{t} - \Delta \cE,~\Delta \cE = \cE(\hat\rho_{SA}(0)) - \cE(\hat\rho_{SA}(\infty)), \\
    \cN_\text{RHP} &= \int_0^\infty \lim_{\varepsilon \to 0^+} \frac{\lVert(\cM_{(t + \varepsilon, t)} \otimes \dI)\ketbra{\Psi}{\Psi}\rVert_1 - 1}{\varepsilon} \dd{t},~\cM_{(t + \varepsilon, t)} \equiv \cM_{t + \varepsilon} \cM_t^{-1},
\end{align}
where $\cS$ is the von Neumann entropy taken in base $2$, and $\sqrt{J}$ is a distance in the strict mathematical sense~\cite{Virosztek2021, Sra2021}; its revivals admit an interpretation in terms of information backflow~\cite{Megier2021}. $\hat\rho_{SA}$ is the system and a non-interacting ancilla prepared in a maximally entangled state $\ketbra{\Psi}{\Psi}$, such that $\hat\rho_{SA}(t) \equiv (\cM_t \otimes \dI)\ketbra{\Psi}{\Psi}$, and $\cE$ denotes some entanglement measure. Note that $\cM_t$ is invertible only where $G(t) \neq 0$: from Eq.~\eqref{eq:map}, $\det\cM_t = |G(t)|^4$, so the map degenerates at any zero of the propagator. This is not a removable technicality: near a simple zero at $t_0$ one has $\gamma(t) = -\partial_t\ln|G(t)|^2 \sim -2/(t - t_0)$, and since $\Lambda^+(t) + \Lambda^-(t) = \gamma(t)$, both canonical rates diverge to $-\infty$ as $t \to t_0^+ $, with $\max\{0,-\Lambda^+\} + \max\{0,-\Lambda^-\} \sim 2/(t - t_0)$. The integrand of Eq.~\eqref{eq:RHPresult} is therefore not integrable and $\cN_\text{RHP}$ diverges logarithmically. A pseudoinverse regularizes the matrix but does not restore the integral; on a numerical grid one obtains a finite value growing as $\log(1/\Delta t)$, so any quoted $\cN_\text{RHP}$ must be checked against grid refinement.

The integrals are written over $[0,\infty)$ but are truncated numerically at a time $t_\text{max}$ with $\Omega t_\text{max} \gg 1$. How quickly they converge depends on the measure and on the regime.

$\cN_\text{BLP}$ and $\cN_\text{LPP}$ see the map only through $|G(t)|$, and they converge even in the bound-state regime, where $|G(t)|$ does not decay. There the localized level contributes a pure phase of constant modulus, which carries no backflow on its own, while the band term in Eq.~\eqref{eq:G-bound} decays as $t^{-3/2}$ for the semicircular density. The per-window increments $\Delta D_u^{(n)}$ therefore fall off fast enough for $\sum_n\Delta D_u^{(n)}$ to converge.

The remaining three measures also involve $Z_\beta(t)$, whose tail is controlled by the Fermi function. Above the threshold, at any $T > 0$, the canonical rates inherit the $t^{-3/2}$ band-edge decay, and the truncation error falls off as $t_\text{max}^{-1/2}$. For $E = 2\Omega$, $\Gamma = \Omega/2$ it is still of order $6\%$ at $\Omega t_\text{max} = 320$, so the values quoted in Sec.~\ref{sec:results} should be read at the stated $t_\text{max}$ rather than as converged limits. On the threshold, at the point of Fig.~\ref{fig:E_0.5_Gamma_1.0}, all measures are converged within $\Omega t_\text{max} = 40$ at every $T > 0$. Below the threshold, with no localized level, $G(t)$ decays to zero and is eventually dominated by the two band edges. Their beating gives $\gamma(t)$ a periodic part of fixed amplitude (Eqs.~\eqref{eq:G-edges} and~\eqref{eq:gamma-edges}), so that $\gamma(t)$ turns negative once per period $\pi/\Omega$ at all late times. $\cN_\text{BLP}$, $\cN_\text{LPP}$ and $\cN_\text{MSV}$, whose increments are weighted by the shrinking Bloch vectors, still converge, though slowly. $\cN^{(0)}_\text{RHP}$, and with it $\cN_\text{RHP}$, instead grows without bound with $t_\text{max}$, so below the threshold $\cN_\text{RHP}$ is meaningful only at a stated window (Fig.~\ref{fig:E_0.5_Gamma_0.5}).

At exactly $T = 0$ the sharp Fermi edge leaves $\Lambda^\pm(t) \sim t^{-1}$. $\cN_\text{MSV}$, $\cN_\text{RHPE}$ and $\cN_\text{RHP}$ then grow logarithmically with $t_\text{max}$, while $\cN_\text{BLP}$ and $\cN_\text{LPP}$ remain finite.

\subsection{BLP} \label{app:BLP}

We will first prove that the optimal pair of input states should be antipodal, as shown, differently and more generally, by~\citet{Wissmann2012}. Taking two generic density matrices as inputs
\begin{equation}
    \rho_1(0) =
    \begin{pmatrix}
    a & b \\
    b^* & 1 - a
    \end{pmatrix},~
    \rho_2(0) =
    \begin{pmatrix}
    c & d \\
    d^* & 1 - c
    \end{pmatrix},
\end{equation}
with $a, c \in [0, 1]$, $|b|^2 \leq a(1 - a)$ and $|d|^2 \leq c(1 - c)$, the trace distance is
\begin{equation}
    D(\rho_1(t), \rho_2(t)) = \sqrt{(a - c)^2|G(t)|^4 + |b - d|^2|G(t)|^2}.
\end{equation}
Let us define $\zeta \equiv (a - c)^{2} \ge 0$ and $\eta \equiv |b - d|^{2} \ge 0$,
and write the objective functional as
\begin{equation}
    \cB[\zeta,\eta] \equiv \int_0^\infty \max\left\{0, \partial_{t}\left(|G(t)| \sqrt{\zeta|G(t)|^{2}+\eta}\right)\right\} \dd{t},~\cN_\text{BLP} = \max_{\{\zeta, \eta\}} \cB[\zeta,\eta].
    \label{eq:blp-functional}
\end{equation}
Under a rescaling $(\zeta,\eta) \to (\lambda\zeta,\lambda\eta)$ with
$\lambda \ge 0$, the square root in Eq.~\eqref{eq:blp-functional} acquires an
overall factor $\sqrt{\lambda}$, which passes through both $\partial_{t}$ and
$\max\{0,\cdot\}$, so that
\begin{equation}
    \cB[\lambda\zeta, \lambda\eta] = \sqrt{\lambda}\cB[\zeta,\eta].
    \label{eq:blp-scaling}
\end{equation}
Since $\cB$ is therefore non-decreasing along every ray from the origin, and strictly increasing wherever $\cB > 0$, the maximum occurs on the Pareto boundary of the feasible set
$(\zeta,\eta)$. Positivity of $\rho_1, \rho_2$ requires that $|b|^2 \leq a(1 - a)$ and $|d|^2 \leq c(1 - c)$, so that
\begin{equation}
    |b - d| \leq |b| + |d| \leq \sqrt{a(1 - a)} + \sqrt{c(1 - c)}.
\end{equation}
We can saturate both inequalities by first selecting opposite phases for $b$ and $d$, such that
\begin{equation}
    |b - d| = |b| + |d|,
\end{equation}
and that the coherence is maximal, such that
\begin{equation}
    |b| + |d| = \sqrt{a(1 - a)} + \sqrt{c(1 - c)},
\end{equation}
and so
\begin{equation}
    |b - d| = \sqrt{a(1 - a)} + \sqrt{c(1 - c)}.
    \label{eq:constraint}
\end{equation}
Now, we fix $\zeta = (a - c)^2$, and we maximize $\sqrt{a(1 - a)} + \sqrt{c(1 - c)}$. The function $\sqrt{x(1 - x)}$ peaks at $x = 1/2$, and given that $(a - c)^2$ is fixed, the maximum is achieved when
\begin{equation}
    a = \frac{1}{2}(1 + \sqrt{\zeta}),~c = \frac{1}{2}(1 - \sqrt{\zeta}).
\end{equation}
Substituting these into Eq.~\eqref{eq:constraint}, we obtain
\begin{align}
    |b - d| &= \sqrt{\frac{1}{2}(1 + \sqrt{\zeta})\left(1 - \frac{1}{2}(1 + \sqrt{\zeta})\right)} + \sqrt{\frac{1}{2}(1 - \sqrt{\zeta})\left(1 - \frac{1}{2}(1 - \sqrt{\zeta})\right)} \nonumber \\
    &= \sqrt{\frac{1}{2}(1 + \sqrt{\zeta})\frac{1}{2}(1 - \sqrt{\zeta})} + \sqrt{\frac{1}{2}(1 - \sqrt{\zeta})\frac{1}{2}(1 + \sqrt{\zeta})} \nonumber \\
    &= \frac{1}{2}\sqrt{1 - \zeta} + \frac{1}{2}\sqrt{1 - \zeta} \nonumber \\
    &= \sqrt{1 - \zeta}.
\end{align}
Since $\eta = |b - d|^2$, this gives
\begin{equation}
    \eta = 1 - \zeta,
\end{equation}
such that the Pareto front is defined as
\begin{equation}
    \zeta + \eta = 1,
\end{equation}
which implies that the BLP optimum is attained by a pair of antipodal states.

We can parametrize the frontier as $u \equiv \eta = 1 - \zeta \in [0,1]$, so that
\begin{equation}
    D_u(t)^2 \equiv (1-u)|G(t)|^4 + u|G(t)|^2,
\end{equation}
which is a convex combination of $|G(t)|^4$ and $|G(t)|^2$, with the two endpoints corresponding to
\begin{align}
    &u = 0: D_0 = |G|^2~\text{Polar pair}~\ketbra{0}{0},~\ketbra{1}{1}, \\
    &u = 1: D_1 = |G|~\text{Equatorial pair}~\ketbra{+}{+},~\ketbra{-}{-}.
\end{align}
Consequently, the BLP measure reduces to a one-parameter problem
\begin{equation}
    \cN_\text{BLP} = \max_{u \in [0, 1]} \cB[1 - u, u].
\end{equation}

In our case, we further tighten the values that $u$ can take. If we look at
\begin{equation}
    \partial_t (D_u(t)^2) = 2D_u(t)\partial_t D_u(t) = 2|G(t)|\partial_t|G(t)|(2(1 - u)|G(t)| + u),
\end{equation}
we note that the terms inside the bracket $2(1 - u)|G(t)| + u \geq 0~\forall u \in [0, 1]$. Since $D_u(t) \geq 0~\forall u \in [0, 1]$ as well, we have that
\begin{equation}
    \sgn(\partial_t D_u(t)) = \sgn(\partial_t|G(t)|)~\forall u \in [0, 1].
\end{equation}
Thus the set of time-intervals on which $D_u(t)$ increases does not depend on the pair of states: it is precisely the set of intervals on which $|G(t)|$ rises.

Now, let us suppose we have $n$ rising windows which we label as intervals $[t^{(n)}_a, t^{(n)}_b]$, where $|G(t^{(n)}_a)|$ increases to $|G(t^{(n)}_b)|$ such that it is continuous and monotonic during that window. Defining $\Delta D_u^{(n)} \equiv D_u(t^{(n)}_b) - D_u(t^{(n)}_a)$, we have that
\begin{equation}
    \cB[1 - u, u] = \sum_{n} \Delta D_u^{(n)}.
\end{equation}
Selecting the endpoints within each window, we get
\begin{align}
    &u = 0: \Delta D_0^{(n)} = |G(t^{(n)}_b)|^2 - |G(t^{(n)}_a)|^2, \\
    &u = 1: \Delta D_1^{(n)} = |G(t^{(n)}_b)| - |G(t^{(n)}_a)|.
\end{align}
Thus, the ratio between polar and equatorial input states is $\Delta D_0^{(n)}/\Delta D_1^{(n)} = |G(t^{(n)}_b)| + |G(t^{(n)}_a)|$. Therefore, whenever $|G(t^{(n)}_b)| + |G(t^{(n)}_a)| > 1$, e.g.\ when a bound state keeps $|G(t)|$ near $1$, polar input states dominate, whereas whenever $|G(t^{(n)}_b)| + |G(t^{(n)}_a)| < 1$, equatorial input states dominate.

The final step of the proof is to show that either $u = 0$ or $u = 1$ maximizes the BLP measure in our case, and not any other value of $u \in (0, 1)$. In general, writing $G_a \equiv |G(t^{(n)}_a)|$ and $G_b \equiv |G(t^{(n)}_b)|$ for brevity, for any $u$, we have that
\begin{align}
    \Delta D_u^{(n)} &= \sqrt{(1 - u)G_b^4 + u G_b^2} - \sqrt{(1 - u)G_a^4 + u G_a^2} \nonumber \\
    &= G_b\sqrt{(1 - u)G_b^2 + u} - G_a\sqrt{(1 - u)G_a^2 + u} \nonumber \\
    &= G_b\sqrt{G_b^2 + u(1 - G^2_b)} - G_a\sqrt{G_a^2 + u(1 - G^2_a)},
\end{align}
and
\begin{equation}
    \nu(u) \equiv \partial_u \Delta D_u^{(n)} = \frac{G_b(1 - G_b^2)}{2\sqrt{G_b^2 + u(1-G_b^2)}} - \frac{G_a(1 - G_a^2)}{2\sqrt{G_a^2 + u(1-G_a^2)}}.
\end{equation}
At $u = 0$, we get
\begin{equation}
    \nu(0) = \frac{1 - G_b^2}{2} - \frac{1 - G_a^2}{2} = \frac{1}{2}(G_a^2 - G_b^2) < 0,
\end{equation}
so $\Delta D_u^{(n)}$ initially decreases with $u$. Now,
\begin{equation}
    \nu(u) \geq 0 \iff \frac{G_b(1 - G_b^2)}{\sqrt{G_b^2 + u(1-G_b^2)}} \geq \frac{G_a(1 - G_a^2)}{\sqrt{G_a^2 + u(1-G_a^2)}}.
\end{equation}
Squaring both sides and cross-multiplying by the denominators results in
\begin{equation}
    \nu(u) \geq 0 \iff G_b^2(1 - G_b^2)^2(G_a^2 + u(1 - G_a^2)) - G_a^2(1 - G_a^2)^2(G_b^2 + u(1 - G_b^2)) \geq 0.
\end{equation}
This implies that the difference between the two sides is affine in $u$, and so $\nu(u)$ changes sign at most once on the interval $[0, 1]$. Hence, if $\nu(u) \leq 0$ throughout the interval $[0, 1]$, such that $\Delta D_u^{(n)}$ is non-increasing, the maximum lies at $u = 0$. Otherwise, $\nu(u)$ changes sign once, from negative to positive, so $\Delta D_u^{(n)}$ has a single interior minimum and the maximum is at an endpoint: either $u = 0$ (when $G_a + G_b > 1$) or $u = 1$ (when $G_a + G_b < 1$).

The maximizing endpoint is therefore fixed within each window by the sign of $G_a + G_b - 1$, which need not be the same in every window: $|G(0)| = 1$, while at long times $|G(t)| \to R_b$, so windows early and late in the evolution may favor different pairs. Since $\cN_\text{BLP}$ requires a single $u$ maximizing $\sum_n \Delta D_u^{(n)}$, we obtain in general the lower bound
\begin{equation}
    \cN_\text{BLP} \geq \max_{P \in \{|G(t)|, |G(t)|^2\}} \int_0^\infty \max\{0, \partial_t P \} \dd{t},
    \label{eq:blpbound}
\end{equation}
with equality whenever one endpoint is optimal in every window.

A sufficient condition for equality, for $\Gamma < 2\Omega$ where a single localized level exists, follows from Eq.~\eqref{eq:G-bound}. Since $\varphi \geq 0$ and $\int \varphi + \sum_b R_b = 1$, the band contribution to $G(t)$ is bounded in modulus by $1 - R_b$, so $|G(t)| \geq 2R_b - 1$ for all $t$; if $R_b \ge 3/4$ then $|G(t)| \geq 1/2$ throughout, every rising window has $G_a + G_b > 1$, and Eq.~\eqref{eq:blpbound} is saturated by the polar pair. The operational criterion, $\min_t|G(t)| \ge 1/2$, is considerably weaker. Scanning $E\in[0,3\Omega]$ and $\Gamma/\Omega\in\{0.2,0.5,1,1.5\}$ we find that a single endpoint is optimal almost everywhere: deep in the bound-state regime the polar pair wins in every window, while for small $|E|$, where $|G|$ decays through small values, the equatorial pair wins in every window. Windows of both types coexist only in a narrow band straddling the crossover (for $\Gamma = \Omega$, $0.65 \lesssim E/\Omega \lesssim 0.73$, where $\min_t|G(t)|$ passes through $1/2$), and there Eq.~\eqref{eq:blpbound} is a strict inequality.

\subsection{LPP}

Our next step is to evaluate the LPP measure. Computing $\mathbf{A}$, we find that
\begin{equation}
    \mathbf{A} =
    \begin{pmatrix}
    \Re(G(t)) & \Im(G(t)) & 0 \\
    -\Im(G(t)) & \Re(G(t)) & 0 \\
    0 & 0 & |G(t)|^2
    \end{pmatrix},~\det\mathbf{A} = |G(t)|^4,
\end{equation}
such that we simply get
\begin{equation}
    \cN_\text{LPP} = \int_0^\infty \max\left\{0, \partial_t|G(t)|^4 \right\} \dd{t}.
\end{equation}

\subsection{MSV} \label{app:MSV}

For $F = 0$ the pair dependence of $\cN_\text{MSV}$ reduces to a small number of parameters. We write each input in Bloch form, $\hat \rho_i(0) = (\dI + \bm v_i\cdot\hat{\bm\sigma})/2$, with $\bm v_i = (a_i\cos\phi_i,a_i\sin\phi_i,w_i)$ and $a_i^2 + w_i^2 \le 1$, where $\hat\sigma^z = \ketbra{1}{1} - \ketbra{0}{0}$, so that $w_i = 2\rho_{11} - 1$ and $a_ie^{\ii\phi_i} = 2\rho_{10}$. By Eq.~\eqref{eq:map} the coherences are simply multiplied by $G(t)$ and the populations obey Eq.~\eqref{eq:rho11}, so that
\begin{align}
    a_i(t) &= |G(t)|a_i(0), \\
    \phi_i(t) &= \phi_i(0) + \arg G(t), \\
    w_i(t) &= |G(t)|^2 w_i(0) + s(t), \\
    s(t) &\equiv |G(t)|^2 + 2Z_\beta(t)-1.
\end{align}
The shift $s$ is common to both states and carries all of the bath dependence. Since the rotation acts equally on both inputs and $J$ is unitarily invariant, $\arg G$ drops out and the azimuths enter only through their difference $\phi = \phi_2 - \phi_1$.
Defining
\begin{align}
    r_i &= |\bm v_i(t)|, \\
    m &= \frac{|\bm v_1(t) + \bm v_2(t)|}{2},
\end{align}
we get
\begin{equation}
    J(t) = h\left(\frac{1 + m}{2}\right) - \frac{1}{2}\left[h\left(\frac{1 + r_1}{2}\right) + h\left(\frac{1 + r_2}{2}\right)\right].
    \label{eq:Jgeneral}
\end{equation}
The optimization in $\cN_\text{MSV}$ is therefore a compact five-parameter problem in $(a_1,w_1,a_2,w_2,\phi)$. Three two-parameter families admit closed forms, and since $\cN_\text{MSV}$ is a maximum over pairs, each of them provides a lower bound.

\emph{(a) Pairs differing only in their populations.} Both Bloch vectors lie on the $z$ axis, $a_1 = a_2 = 0$, so both states are diagonal in the occupation basis. Their occupations evolve as $p_i(t) = |G(t)|^2p_i(0) + Z_\beta(t)$, so that, in the variables of Sec.~\ref{sec:qjsd}, $c(t) = |G(t)|^2c_0 + Z_\beta(t)$ and $\Delta(t) = |G(t)|^2\Delta_0$. With $c_0$ and $\Delta_0$ the initial mean occupation and gap, constrained by $0 \le c_0\pm\Delta_0/2 \le 1$, we get
\begin{equation}
    J_\text{pop} = h\left(|G|^2 c_0 + Z_\beta\right) - \frac{1}{2}\left[ h\left(|G|^2\left(c_0 + \frac{\Delta_0}{2}\right) + Z_\beta\right) + h\left(|G|^2\left(c_0 - \frac{\Delta_0}{2}\right) + Z_\beta\right) \right].
    \label{eq:Jpop}
\end{equation}
The polar pair is the corner $c_0 = 1/2$, $\Delta_0 = 1$; the rest of the family consists of pairs of \emph{mixed} states sitting at arbitrary heights on the axis. The gap decays as $|G|^2$ for every member of the family, so $D = |\Delta|$ differs between pairs only by the constant factor $\Delta_0$: all pairs increase on the same set of intervals, those on which $|G|$ rises. This is the property that makes the window-by-window comparison of Appendix~\ref{app:BLP} possible, and it fails for $J$. By Eq.~\eqref{eq:Jdot} the sign of $\partial_t J$ also involves $\partial_c J\partial_t c$, whose sign is that of $(c - 1/2)\partial_t c$ by Theorem~\ref{thm:qjsd}(ii), and $c(t) = |G|^2c_0 + Z_\beta$ depends on the pair through $c_0$. Different members of the family therefore rise on different sets of intervals, and no single endpoint can be argued to dominate in all of them.

\emph{(b) Pairs differing only in their coherences.} Both states sit at the same height, $w_1 = w_2 = w$, with opposite coherences, $a_1 = a_2 = a$ and $\phi = \pi$. Their populations therefore coincide at all times, at the common value fixed by $z \equiv |G|^2 w + s$, while their transverse components remain equal and opposite. The two Bloch radii coincide, $r_1 = r_2 = r = \sqrt{|G|^2 a^2 + z^2}$, and the transverse parts cancel in the mean, leaving $m = |z|$, so that Eq.~\eqref{eq:Jgeneral} collapses to
\begin{equation}
    J_\text{coh} = h\left(\frac{1 + |z|}{2}\right) - h\left(\frac{1 + r}{2}\right).
\end{equation}
Since $r \ge |z|$ and $h((1 + x)/2)$ is decreasing, $J_\text{coh} \ge 0$, with equality only when $|G|a = 0$. The equatorial pair is $a = 1$, $w = 0$.

\emph{(c) Antipodal pairs.} Taking $\bm v_2 = -\bm v_1$ with $\bm v_1 = (a,0,w)$ reproduces the family used for the trace distance above, with $a = \sin\theta$ and $w = \cos\theta$ interpolating between the polar ($\theta = 0$) and equatorial ($\theta = \pi/2$) pairs. Here the mean Bloch vector is pair independent, $m = |s(t)|$, and
\begin{equation}
    J_\text{anti} = h\left(\frac{1 + |s|}{2}\right) - \frac{1}{2}\left[h\left(\frac{1 + r_1}{2}\right) + h\left(\frac{1 + r_2}{2}\right)\right],~r_{1,2} = \sqrt{|G|^2a^2 + \left(s \pm |G|^2w\right)^2}.
\end{equation}

Collecting these,
\begin{equation}
    \cN_\text{MSV} \ge \max_{P\in\{J_\text{pop}(t), J_\text{coh}(t), J_\text{anti}(t)\}} \int_0^\infty \max\left\{0,\partial_t\sqrt{P}\right\}\dd{t},
    \label{eq:msvbound}
\end{equation}
the analogue of Eq.~\eqref{eq:blpbound}. Unlike the trace-distance case, however, we have no proof that any of these families contains the optimum: both ingredients of the argument of Appendix~\ref{app:BLP} fail here, since $\sqrt{J}$ obeys no scaling relation of the form Eq.~\eqref{eq:blp-scaling}, and the set of intervals on which the measure rises is pair dependent.

Table~\ref{tab:msvfamilies} compares the three families with each other, with the polar pair, and with a multistart optimization over all five parameters. Within the antipodal family~(c) the polar pair remains optimal in every case examined, so the trace-distance reduction carries over there; what improves on it is leaving that family. Pairs of mixed diagonal states, family~(a), exceed the polar value by up to $35\%$ at low temperature (for $E = 2\Omega$, $\Gamma = \Omega/2$ at $T = 0.02\Omega$ the optimum is $(c_0,\Delta_0) \simeq (0.11,0.22)$), while the coherence pairs of family~(b), lifted off the equator, also exceed it and give the best of the three bounds in the weak-revival regime, $E = \Omega/2$ at $T = \Omega$. Both exploit the dependence on the mean occupation that the trace distance does not have. The unrestricted optimum exceeds the best family value by at most $7\%$, in the weak-signal case $E = \Omega/2$, $T = \Omega$, where the optimal pair differs in both its populations and its coherences, and by about $1.4\%$ at $E = \Omega/2$, $T = 0.3\Omega$; elsewhere the two agree to the resolution of the scan. The values quoted in the main text are therefore lower bounds on $\cN_\text{MSV}$. This suffices for the argument of Sec.~\ref{sec:qjsd}, which requires a strictly positive value at finite temperature together with the exactly vanishing value at $\beta \to 0$, where every pair gives zero.

\begin{table}[t]
    \caption{
    Revivals of $\sqrt{J}$ maximized within each family of Appendix~\ref{app:MSV}, compared with the polar pair, and with a multistart Nelder--Mead optimization over all five parameters of Eq.~\eqref{eq:Jgeneral} (last column, $40$ restarts). Every entry is a lower bound on $\cN_\text{MSV}$. The families are maximized on a grid and the free column by a local search, so the latter sits marginally below the former in two rows; the discrepancy is at most $0.06\%$. Time integration to $\Omega t_\text{max} = 100$ on $2001$ points.
    }
    \label{tab:msvfamilies}
    \begin{ruledtabular}
    \begin{tabular}{cccccccc}
        $E/\Omega$ & $\Gamma/\Omega$ & $T/\Omega$ & \multicolumn{1}{c}{polar} & \multicolumn{1}{c}{(a) populations} & \multicolumn{1}{c}{(b) coherences} & \multicolumn{1}{c}{(c) antipodal} & \multicolumn{1}{c}{free} \\ \colrule
        0.5 & 1.0 & 0.02 & 0.00944 & 0.00944 & 0.00777 & 0.00944 & 0.00944 \\
        0.5 & 1.0 & 0.30 & 0.00234 & 0.00234 & 0.00214 & 0.00234 & 0.00237 \\
        0.5 & 1.0 & 1.00 & 0.00012 & 0.00014 & 0.00017 & 0.00012 & 0.00018 \\
        1.0 & 0.5 & 0.02 & 0.09978 & 0.10681 & 0.09756 & 0.09978 & 0.10680 \\
        1.0 & 0.5 & 0.30 & 0.07719 & 0.07719 & 0.05700 & 0.07719 & 0.07719 \\
        1.0 & 0.5 & 1.00 & 0.08437 & 0.08437 & 0.06084 & 0.08437 & 0.08437 \\
        2.0 & 0.5 & 0.02 & 0.11356 & 0.15377 & 0.15065 & 0.11356 & 0.15368 \\
        2.0 & 0.5 & 0.30 & 0.11375 & 0.11375 & 0.08990 & 0.11375 & 0.11375 \\
        2.0 & 0.5 & 1.00 & 0.11572 & 0.11572 & 0.07815 & 0.11572 & 0.11572 \\
    \end{tabular}
    \end{ruledtabular}
\end{table}

\subsection{RHPE}

Now, we calculate the RHPE measure. Using the concurrence
$\cC$~\cite{Wootters1998} as our entanglement measure, and the Bell
state $\ket{\Psi} = (\ket{00} + \ket{11})/\sqrt{2}$ as the maximally entangled
state, we have that
\begin{equation}
    \cC(\hat\rho_{SA}(t)) = \max\left\{0, |G(t)| - \sqrt{Z_{\beta}(t)(1 - |G(t)|^{2} - Z_{\beta}(t))} \right\},
    \label{eq:concurrence}
\end{equation}
such that
\begin{equation}
    \cN_\text{RHPE} = \int_{0}^\infty|\partial_{t}\cC(\hat\rho_{SA}(t))|\dd{t} - [\cC(\hat\rho_{SA}(0)) - \cC(\hat\rho_{SA}(\infty))],~\cC(\hat\rho_{SA}(0)) = 1,
\end{equation}
which, importantly, contains the temperature dependence via $Z_{\beta}(t)$, unlike the BLP and LPP measures.

\subsection{RHP} \label{app:RHP}

Finally, we investigate the RHP measure. In the limit $\varepsilon \to 0$, the solution of Eq.~\eqref{eq:master_equation} tends to $\cM_{(t + \varepsilon, t)} = e^{\cL \varepsilon}$~\cite{Rivas2010}. As such, we only need to expand it in the trace norm up to first order to calculate
\begin{equation}
    \lVert(\dI + (\cL \otimes \dI)\varepsilon)\ketbra{\Psi}{\Psi}\rVert_1.
\end{equation}
Computing the vectorized $\cL$~\cite{Landi2024}, with $\tilde H_S(t)$ as in Sec.~\ref{sec:master_equation}, we obtain
\begin{align}
    \cL = -\ii( \dI \otimes \tilde H_S(t) - \tilde H_S(t)^T \otimes \dI ) &+ \Lambda^+(t) \left( (\hat d^\dagger)^* \otimes \hat d^\dagger - \frac{1}{2} \left( \dI \otimes \hat d\hat d^\dagger + (\hat d\hat d^\dagger)^T \otimes \dI \right) \right) \nonumber \\ &+ \Lambda^-(t) \left( \hat d^* \otimes \hat d - \frac{1}{2} \left( \dI \otimes \hat d^\dagger \hat d + (\hat d^\dagger \hat d)^T \otimes \dI \right) \right),
\end{align}
which subsequently can be rewritten in a matrix form
\begin{equation}
    \cL =
    \begin{pmatrix}
    -\Lambda^+(t) & 0 & 0 & \Lambda^-(t) \\
    0 & -\ii H(t) - \frac{\gamma(t)}{2} & 0 & 0 \\
    0 & 0 & \ii H(t) - \frac{\gamma(t)}{2} & 0 \\
    \Lambda^+(t) & 0 & 0 & -\Lambda^-(t)
    \end{pmatrix},
\end{equation}
where $\gamma(t) \equiv \Lambda^+(t) + \Lambda^-(t)$. Computing the Choi matrix~\cite{Choi1975} $\chi_\varepsilon \equiv (\dI + (\cL \otimes \dI)\varepsilon)\ketbra{\Psi}{\Psi}$, we get
\begin{equation}
    \chi_\varepsilon = \frac{1}{2}
    \begin{pmatrix}
    1 - \varepsilon\Lambda^+(t) & 0 & 0 & 1 - \varepsilon\left(\frac{\gamma(t)}{2} - \ii H(t)\right) \\
    0 & \varepsilon\Lambda^-(t) & 0 & 0 \\
    0 & 0 & \varepsilon\Lambda^+(t) & 0 \\
    1 - \varepsilon\left(\frac{\gamma(t)}{2} + \ii H(t)\right) & 0 & 0 & 1 - \varepsilon\Lambda^-(t)
    \end{pmatrix}.
\end{equation}
This is a block matrix with the middle block eigenvalues $\varepsilon\Lambda^+(t)/2$ and $\varepsilon\Lambda^-(t)/2$. To calculate the corner block eigenvalues, one can use the trace and the determinant of this matrix. The trace is equal to
\begin{equation}
    \Tr{\chi_\varepsilon^{(c)}} = 1 - \frac{\varepsilon\gamma(t)}{2},
\end{equation}
while the determinant reads
\begin{align}
    \det\chi_\varepsilon^{(c)} &= \frac{1}{4}\left( (1 - \varepsilon\Lambda^+(t))(1 - \varepsilon\Lambda^-(t)) - \left(1 - \varepsilon\left(\frac{\gamma(t)}{2} + \ii H(t)\right)\right) \left(1 - \varepsilon\left(\frac{\gamma(t)}{2} - \ii H(t)\right) \right) \right) \nonumber \\
    &= \frac{1}{4}\left( 1 - \varepsilon\Lambda^+(t) - \varepsilon\Lambda^-(t) + \cO(\varepsilon^2) - \left(1 - \varepsilon\left(\frac{\gamma(t)}{2} + \ii H(t)\right) - \varepsilon\left(\frac{\gamma(t)}{2} - \ii H(t)\right) + \cO(\varepsilon^2)\right) \right) \nonumber \\
    &= \frac{1}{4}\left( 1 - \varepsilon\gamma(t) - \left(1 - \varepsilon\gamma(t)\right) \right) + \cO (\varepsilon^2) \nonumber \\
    &= \cO (\varepsilon^2),
\end{align}
which vanishes to first order in $\varepsilon$. Thus, the eigenvalues of the corner block are $1 - \varepsilon\gamma(t)/2$ and $\cO (\varepsilon^2)$. The spectrum is therefore
\begin{equation}
    \{\lambda_i\}_i = \left\{\frac{\varepsilon\Lambda^+(t)}{2}, \frac{\varepsilon\Lambda^-(t)}{2}, \cO (\varepsilon^2), 1 - \frac{\varepsilon\gamma(t)}{2}\right\}.
\end{equation}
The intermediate map $\cM_{(t + \varepsilon,t)}$ is CP if and only if $\chi_\varepsilon \succeq 0$~\cite{Choi1975, Rivas2010}. The only eigenvalues which can be negative are $\varepsilon\Lambda^+(t)/2$ and $\varepsilon\Lambda^-(t)/2$. Note that the eigenvalue $\cO (\varepsilon^2)$ contributes $\cO (\varepsilon^2)$ to $\lVert\chi_\varepsilon\rVert_1$, and therefore vanishes in the limit $\varepsilon \to 0^+ $. Following the recipe of~\citet{Rivas2010}, and since $\chi_\varepsilon$ is Hermitian and $\Tr{\chi_\varepsilon} = 1$, we evaluate
\begin{align}
    \lVert\chi_\varepsilon\rVert_1 &= \sum_i |\lambda_i| \nonumber \\
    &= \sum_{\lambda_i \geq 0} |\lambda_i| + \sum_{\lambda_i < 0} |\lambda_i| \nonumber \\
    &= \sum_{\lambda_i \geq 0} \lambda_i - \sum_{\lambda_i < 0} \lambda_i \nonumber \\
    &= \left(1 - \sum_{\lambda_i < 0} \lambda_i\right) - \sum_{\lambda_i < 0} \lambda_i \nonumber \\
    &= 1 - 2\sum_{\lambda_i < 0} \lambda_i \nonumber \\
    &= 1 + 2\left(\max\left\{0, -\frac{\varepsilon\Lambda^+(t)}{2}\right\} + \max\left\{0, -\frac{\varepsilon\Lambda^-(t)}{2}\right\}\right) \nonumber \\
    &= 1 + \varepsilon\left(\max\left\{0, -\Lambda^+(t)\right\} + \max\left\{0, -\Lambda^-(t)\right\}\right).
\end{align}
Thus, we have
\begin{align}
    \cN_\text{RHP} &= \int_0^\infty \lim_{\varepsilon \to 0^+} \frac{\left( 1 + \varepsilon\left(\max\left\{0, -\Lambda^+(t)\right\} + \max\left\{0, -\Lambda^-(t)\right\}\right) \right) - 1}{\varepsilon} \dd{t} \nonumber \\
    &= \int_0^\infty \max\left\{0, -\Lambda^+(t)\right\} + \max\left\{0, -\Lambda^-(t)\right\} \dd{t},
    \label{eq:RHPresult}
\end{align}
which is precisely the result given by~\citet{Hall2014}. This expression also contains the temperature dependence in $\Lambda^+(t)$ and $\Lambda^-(t)$, unlike the BLP and LPP measures.

\section{Long-Time Behavior Below the Threshold} \label{app:longtime}

Below the threshold, the numerical results of Sec.~\ref{sec:results} show two features that the exact solution of Appendix~\ref{app:dynamical_map} does not make evident: revivals of $|G(t)|$ that set in ever later as the threshold is approached, and a divisibility measure that grows without bound with $t_\text{max}$. Both follow from the long-time form of the propagator, which we derive here.

For $|E| < E_c$, only the band integral remains in Eq.~\eqref{eq:G-bound}, and its long-time behavior is set by the band edges, where $\varphi(\omega)$ is not smooth: $\cJ_\text{N}$ vanishes there as a square root, $\pi\cJ_\text{N}\to0$ and $\Re\Sigma(\pm\Omega) = \pm\Gamma/2$, so that $\varphi(\omega) \simeq \cJ_\text{N}(\omega)/(E_c \mp E)^2$ near $\omega = \pm\Omega$. Integrating against $e^{-\ii\omega t}$ near each edge gives
\begin{equation}
    G(t) \simeq \frac{\alpha_+e^{-\ii\Omega t} + \alpha_-e^{\ii\Omega t}}{(\Omega t)^{3/2}},~\alpha_\pm = \frac{\Gamma\Omega e^{\pm3\pi\ii/4}}{2\sqrt{2\pi} (E_c \mp E)^2},
    \label{eq:G-edges}
\end{equation}
up to corrections of relative order $(\Omega t)^{-1}$, while the resonance near $E'$ decays exponentially. The two edge terms beat at frequency $2\Omega$,
\begin{equation}
    \gamma(t) \simeq \frac{3}{t} + \frac{4\Omega|\alpha_+\alpha_-|\cos(2\Omega t)}{|\alpha_+|^2 + |\alpha_-|^2 - 2|\alpha_+\alpha_-|\sin(2\Omega t)},
    \label{eq:gamma-edges}
\end{equation}
whose periodic part has a fixed amplitude, with minimum $-4\Omega|\alpha_+\alpha_-|/\big||\alpha_+|^2 - |\alpha_-|^2\big|$. Once the edges dominate, $\gamma(t)$ therefore turns negative once in every period $\pi/\Omega$, for all $t \gtrsim t_\text{first}$ with
\begin{equation}
    t_\text{first} \simeq \frac{6E_c|E|(E_c^2 + E^2)}{\Omega(E_c^2 - E^2)^2},
    \label{eq:tfirst}
\end{equation}
which diverges as $(E - E_c)^{-2}$ at the threshold.

Because each revival is weighted by $|G| \sim t^{-3/2}$, $\cN_\text{BLP}$, $\cN_\text{LPP}$ and $\cN_\text{MSV}$ converge, whereas the negative part of $\gamma(t)$, and with it $\cN^{(0)}_\text{RHP}$, does not decay (Appendix~\ref{app:NM}).

\twocolumngrid
\bibliographystyle{apsrev4-1fixed_with_article_titles_full_names.bst}
\bibliography{ref}

\end{document}